\documentclass[11pt,a4paper]{article}

\usepackage{amsmath, amssymb, amsthm, latexsym, mathrsfs} 
\usepackage{url, color, epsfig, graphics, float, sectsty, setspace}
\usepackage[T1]{fontenc}
\usepackage[bitstream-charter]{mathdesign}

\usepackage[pdftex]{hyperref} 

\newtheorem{definition}{Definition}[section]
\newtheorem{remark}{Remark}
\newtheorem{proposition}{Proposition}[section]
\newtheorem{example}{Example}[section]
\numberwithin{equation}{section}

\newcommand{\cE}{\mathcal{E}}
\newcommand{\cH}{\mathcal{H}}
\newcommand{\cF}{\mathcal{F}}
\newcommand{\cG}{\mathcal{G}}
\newcommand{\pr}{\mathbb{P}}
\newcommand{\re}{\mathbb{R}}
\newcommand{\B}{\mathbb{B}}
\newcommand{\EP}{\mathbb{E}}
\newcommand{\rd}{\textup{d}}
\newcommand{\hM}{\widehat{M}}
\newcommand{\indi}[1]{1\hspace{-.09cm}\textup{\textrm{l}}}

\newcommand{\nn}{\nonumber}

\begin{document}
\title{\bf Randomised Mixture Models for Pricing Kernels}
\author{Andrea Macrina$^{1}$\footnote{Corresponding author. E-mail: andrea.macrina@kcl.ac.uk}\; and Priyanka A. Parbhoo$^{1,\,2}$\\\\
{\normalsize $^{1\,}$Department of Mathematics, King's College London, London WC2R 2LS, UK}\\
{\normalsize $^{2\,}$School of Computational and Applied Mathematics,}\\ {\normalsize University of the Witwatersrand, Private Bag 3, WITS 2050, South Africa}}
\date{9 December 2011}
\maketitle
\vspace{-0.75cm}
\begin{abstract}
\noindent
Numerous kinds of uncertainties may affect an economy, e.g. economic, political, and environmental ones. We model the aggregate impact by the uncertainties on an economy and its associated financial market by randomised mixtures of L\'evy processes. We assume that market participants observe the randomised mixtures only through best estimates based on noisy market information. The concept of incomplete information introduces an element of stochastic filtering theory in constructing what we term ``filtered Esscher martingales''. We make use of this family of martingales to develop pricing kernel models. Examples of bond price models are examined, and we show that the choice of the random mixture has a significant effect on the model dynamics and the types of movements observed in the associated yield curves. Parameter sensitivity is analysed and option price processes are derived. We extend the class of pricing kernel models by considering a weighted heat kernel approach, and develop models driven by mixtures of Markov processes.
\\\vspace{-0.2cm}
\end{abstract}


\section{Introduction}
In this paper, we develop interest rate models that offer consistent dynamics in the short, medium, and long term. Often interest rate models have valid dynamics in the short term, that is to say, over days or perhaps a few weeks. Such models may be appropriate for the pricing of securities with short time-to-maturity. For financial assets with long-term maturities, one requires  interest rate models with plausible long-term dynamics, which retain their validity over years. Thus the question arises as to how one can create interest rate models which are sensitive to market changes over both short and long time intervals, so that they remain useful for the pricing of securities of various tenors. Ideally, one would have at one's disposal interest rate models that allow for consistent pricing of financial instruments expiring within a range of a few minutes up to years, and if necessary over decades. One can imagine an investor holding a portfolio of securities maturing over various periods of time, perhaps spanning several years. Another situation requiring interest rate models that are valid over short and long terms, is where illiquid long-term fixed-income assets need to be replicated with (rolled-over) liquid shorter-term derivatives. Here it is central that the underlying interest rate model possesses consistent dynamics over all periods of time in order to avoid substantial hedging inaccuracy.  Insurance companies, or pension funds, holding liabilities over decades might have no other means but to invest in shorter-term derivatives, possibly with maturities of months or a few years, in order to secure enough collateral for their long-term liabilities reserves. Furthermore, such hedges might in turn need second-order liquid short-term protection, and so forth. Applying different interest rate models validated for the various investment periods, which frequently do not guarantee price and hedging consistency, seems undesirable. Instead, we propose a family of pricing kernel models which may generate interest rate dynamics sufficiently flexible to allow for diverse behaviour over short, medium and long periods of time.

We imagine economies, and their associated financial markets, that are exposed to a variety of uncertainties, such as economic, social, political, environmental, or demographic ones.  We model the degree of impact of these underlying factors on an economy (and financial markets) at each point in time by combinations of continuous-time stochastic processes of different probability laws. When designing interest rate models that are sensitive to the states an economy may take, subject to its response to the underlying uncertainty factors, one may wonder a) how many stochastic factor processes ought to be considered, and b) what is the combination, or mixture, of factor processes determining the dynamics of an economy and its associated financial market.  It is plausible to assume that the number of stochastic factors and their combined impact on a financial market continuously changes over time, and thus that any interest rate model designed in such a set-up is by nature time-inhomogeneous. The recipe used to construct interest-rate models within the framework proposed in this paper can be summarised as follows: 
\begin{enumerate}
\item [(i)] Assume that the response of a financial market to uncertainty is modelled by a family of stochastic processes, e.g. Markov processes.

\item [(ii)] Consider a mixture of such stochastic processes as the basic driver of the resulting interest rate models. 

\item [(iii)] In order to explicitly design interest rate models, apply a method for the modelling of the pricing kernel associated with the economy, which underlies the considered financial market. 

\item [(iv)] Derive the interest rate dynamics directly from the pricing kernel models, or, if more convenient, deduce the interest rate model from the bond price process associated with the constructed pricing kernel. 
\end{enumerate}

The set of stochastic processes chosen to model an economy's response to uncertainty, the particular mixture of those, and the pricing kernel model jointly characterize the dynamics of the derived interest rate model. We welcome these degrees of freedom, for any one of them may abate the shortcoming (or may amplify the virtues) of another. For example, one might be constrained to choose L\'evy processes to model the impact of uncertainty on markets. The fact that L\'evy processes are time-homogeneous processes with independent increments, might be seen as a disadvantage for modelling interest rates for long time spans. However, a time-dependent pricing kernel function may later introduce time-inhomogeneity in the resulting interest rate model. The choice of a certain set of stochastic processes implicitly determines a particular joint law of the modelled market response to the uncertainty sources. Although the resulting multivariate law may not coincide well with the law of the combined uncertainty impact, the fact that we can directly model a particular mixture of stochastic processes provides the desirable degree of freedom in order to control the dynamical law of the market's response to uncertainty. In this paper, we consider ``randomised mixing functions'' for the construction of multivariate interest rate models with distinct response patterns to short-, medium-, and long-term uncertainties. Having a randomised mixing function enables us to introduce the concept of ``partially-observable mixtures'' of stochastic processes. We take the view that market agents cannot fully observe the actual combination of processes underlying the market. Instead they form best estimates of the randomised mixture given the information they possess; these estimates are continuously updated as time elapses. This feature introduces a feedback effect in the constructed pricing models.   

The reason why we prefer to propose pricing kernel models in order to generate the dynamics of interest rates, as opposed to modelling the interest rates directly, is that the modelling of the pricing kernel offers an integrated approach to equilibrium asset pricing in general (see Cochrane \cite{coc}, Duffie \cite{duf}), including risk management and thus the quantification of risk involved in an investment. The pricing kernel includes the \textit{quantified} total response to the uncertainties affecting an economy or, in other words, the risk premium asked by an investor as an incentive for investing in risky assets. In this work we first consider a particular family of pricing kernel models, namely the Flesaker-Hughston class (see Flesaker \& Hughston \cite{fh}, Hunt \& Kennedy \cite{hk}, Cairns \cite{cai}, Brigo \& Mercurio \cite{bm}). Since our goal in this paper is to primarily introduce a framework capable of addressing issues arising in interest rate modelling over short to long term time intervals, we apply our ideas first to the Flesaker-Hughston class of pricing kernels. We conclude the paper by introducing randomised weighted heat kernel models, along the lines of Akahori {\it et al}.~\cite{ahtt} and Akahori \& Macrina \cite{am}, which extend the class of pricing kernels developed in the first part of this paper.   


\section{Randomised Esscher martingales}
We begin by introducing the mathematical tools that we shall use to construct pricing kernel models based on randomised mixtures of L\'evy processes.
We fix a probability space $(\Omega, \cF, \pr)$ where $\pr$ denotes the real probability measure.

\begin{definition}\label{RandEssMart}
Let $\{L_t\}_{t\geq 0}$ be an $n$-dimensional L\'evy process with independent components, and let $X:\Omega\rightarrow\re^m$ be an independent, $m$-dimensional vector of random variables. For $t,u\in\re_+$, the process $\{M_{tu}(X)\}$ is defined by
\begin{equation}\label{martingale}
M_{tu}(X)=\frac{\exp{\left(h(u, X)L_t\right)}}{\EP\left[\exp{\left(h(u,X)L_t\right)}\,|\,X\right]},
\end{equation}
where the function $h:\re_+\times\re^m\rightarrow\re^n$ is chosen such that $\EP\left[\,|\,M_{tu}(X)\,|\,\right]<\infty$ for all $t\in \re_+$.
\end{definition}

\begin{proposition}
Let the filtration $\{\cH_t\}_{t\geq 0}$ be given by $\cH_t=\sigma\left(\{L_s\}_{0\le s\le t},X\right)$.
Then the process $\{M_{tu}(X)\}$ is an $(\{\cH_t\},\pr)$-martingale.
\end{proposition}
We note that $X$ is $\cH_0$-measurable and therefore, that $\{\cH_t\}$ is an initial enlargement of the natural filtration of $\{L_t\}$ by the random variable $X$. Furthermore, $M_{0u}(X)=1$ and $M_{tu}(X)> 0$ for all $t,u\in\re_+$.

\medskip  
\begin{proof}
The condition that $\EP\left[\,|\,M_{tu}(X)\,|\,\right]$ be finite for all $0\le t<\infty$ is ensured by definition. It remains to be shown that 
\begin{equation}
\EP\left[M_{tu}(X)\,\vert\,\cH_s\right]=M_{su}(X)
\end{equation}
for all $0\le s\le t <\infty$. We observe that the denominator in (\ref{martingale}) is $\cH_0$-measurable so that we can write
\begin{equation}
\EP\left[M_{tu}(X)\,|\,\cH_s\right] = \frac{\EP\left[\exp{\left(h(u, X)L_t\right)}\,|\,\cH_s\right]}{\EP \left[\exp{\left(h(u,X)L_t\right)}\,|\,X\right]}. 
\end{equation}
Next we expand the right-hand-side of the above equation to obtain
\begin{equation}
\frac{\EP\left[\exp\left[h(u,X)(L_t-L_s)\right]\exp\left[h(u,X)L_s\right]\,\vert\,\cH_s\right]}{\EP\left[\exp\left[h(u,X)(L_t-L_s)\right]\exp\left[h(u,X)L_s\right]\,\vert\,X\right]}.
\end{equation}
Given $X$, the expectation in the denominator factorizes since $L_t-L_s$ is independent of $L_s$. In addition, the factor $\exp[h(u,X)L_s]$ is $\cH_s$-measurable so that we may write
\begin{equation}\label{fact-mart}
\EP\left[M_{tu}(X)\,|\,\cH_s\right] = \frac{\exp\left[h(u,X)L_s\right]}{\EP\left[\exp\left[h(u,X)L_s\right]\,\vert\,X\right]}\ \frac{\EP\left[\exp\left[h(u,X)(L_t-L_s)\right]\,\vert\,\cH_s\right]}{\EP\left[\exp\left[h(u,X)(L_t-L_s)\right]\,\vert\,X\right]}.
\end{equation}
Since the increment $L_t-L_s$ and $X$ are independent of $L_s$, the $\cH_s$-conditional expectation reduces to an expectation conditional on $X$. Thus, equation (\ref{fact-mart}) simplifies to
\begin{equation}
\EP\left[M_{tu}(X)\,|\,\cH_s\right] = \frac{\exp\left[h(u,X)L_s\right]}{\EP\left[\exp\left[h(u,X)L_s\right]\,\vert\,X\right]},
\end{equation}
which is $M_{su}(X)$. 
\end{proof}
\medskip

\indent We call the family of processes $\{M_{tu}(X)\}$ parameterised by $u\in\re_+$ the ``randomised Esscher martingales'' (see Gerber \& Shiu \cite{gs} and Yao \cite{yao} for details on the Esscher transform). The randomization is produced by $h(u,X)$ which we call the ``random mixer''. 
\begin{example}\label{BEM}
{\rm
Let $\{W_t\}_{t\geq 0}$ be a standard Brownian motion that is independent of $X$, and set $L_t=W_t$ in Definition \ref{RandEssMart}. Then,  
\begin{equation}\label{examplebm}
M_{tu}(X) = \exp{\left[h(u,X)W_t - \tfrac{1}{2}\,h^2(u,X)t\right]}.
\end{equation}
}
\end{example}

\begin{example}\label{GEM}
{\rm
Let $\{\gamma_t\}_{t\geq 0}$ be a gamma process with rate parameter $m > 0$ and scale parameter $\kappa > 0$. Then $\EP[\gamma_t] = \kappa m t$ and $\textup{Var}[\gamma_t] = \kappa^2 m t$. We assume that
$\{\gamma_t\}$ is independent of $X$. Set $L_t=\gamma_t$ in Definition \ref{RandEssMart}. Then, if $h(u,X) < \kappa^{-1}$, we have
\begin{equation}
M_{tu}(X) = \left[1-\kappa\,h(u,X)\right]^{mt}\exp{\left[h(u,X)\,\gamma_t\right]}.
\end{equation}
}
\end{example}


\section{Filtered Esscher martingales}
In this section we construct a projection of the randomised Esscher martingales that can be interpreted as follows. Let us suppose that the exact combination of L\'evy processes that forms the stochastic basis of the martingale family $\{M_{tu}(X)\}$ is unknown. That is, we may have little knowledge about how much each of the L\'evy processes involved actually contributes to the stochastic evolution of $\{M_{tu}(X)\}$.  The random vector $h(u,X)$ however, can naturally be interpreted as the quantity inside $\{M_{tu}(X)\}$ that determines at time $u$ the random mixture of L\'evy processes driving the martingale family. Given a certain set of information, the actual mixture might not be fully observable, though. This leads us to the following construction that applies the theory of stochastic filtering. For simplicity, we focus on the case where $X$ is a one-dimensional random variable.
\\

We introduce a standard Brownian motion $\{B_t\}_{t\geq 0}$ on $(\Omega, \cF, \pr)$, and define the filtration $\{\cG_t\}$ by
\begin{equation}\label{Gfil}
\cG_t=\sigma\left(\{B_s\}_{0\leq s\leq t},\{L_s\}_{0\leq s\leq t}, X\right),
\end{equation}
where $\{B_t\}$ is taken to be independent of $X$ and $\{L_t\}$. We consider the pair
\begin{align}
 &\rd X_t=0,\label{signal}\\
 &\rd I_t=\ell(t,X)\rd t+\rd B_t,\label{InfoProc}
\end{align}
where $\ell:\re_+\times\re\rightarrow\re$ is a well-defined function. The solution to the signal equation (\ref{signal}) is of course the random variable $X$. In the theory of stochastic filtering, the process $\{I_t\}_{t\geq 0}$ is the so-called observation process. We have
\begin{equation}\label{Inonmarkov}
I_t = \int_0^t \ell(s,X)\rd s + B_t.
\end{equation}
Next, we introduce the filtration $\{\cF_t\}_{t\geq 0}$ defined by
\begin{equation}\label{Ffil}
 \cF_t=\sigma\left(\{I_s\}_{0\le s\le t},\{L_s\}_{0\le s\le t}\right),
\end{equation}
where $\cF_t\subset\cG_t$. The filtration $\{\cF_t\}$ provides full information about the L\'evy process $\{L_t\}$, however it only gives partial information about the random variable $X$. Let us thus consider the stochastic filtering problem defined by
\begin{equation}\label{hatM}
\hM_{tu} = \EP\left[M_{tu}(X)\,|\, \cF_t\right].
\end{equation} 
We emphasize that $X$ is not $\cF_t$-measurable and thus $\{M_{tu}(X)\}$ is not adapted to $\{\cF_t\}$. The filtering problem (\ref{hatM}) is solved in closed form by introducing 
\begin{equation}
\cE_t := \exp{\left(-\int_0^t \ell(s, X)\rd B_s - \tfrac{1}{2}\int_0^t \ell^2(s, X)\rd s\right)},
\end{equation}
where for all $t>0$
\begin{equation}
\mathbb{E}\left[\int_0^t \ell(s,X)^2 \rd s\right] < \infty, 
\end{equation}
and
\begin{equation}
\mathbb{E}\left[\int_0^t \cE_s\, \ell(s, X)^2 \rd s\right] < \infty.
\end{equation}
The process $\{\cE_t\}$ is a $(\{\cG_t\}, \pr)$-martingale (see, e.g., Bain \& Crisan \cite{bc}), and it may be used to define a change-of-measure density martingale from $\pr$ to a new measure $\mathbb{B}$ by setting
\begin{equation}
\frac{\rd \B}{\rd \pr}\bigg|_{\cG_t} = \mathcal{E}_t.
\end{equation}
The $\mathbb{B}$-measure is characterised by the fact that $\{I_t\}$ is an $(\{\cF_t\}, \mathbb{B})$-Brownian motion. The Kallianpur-Striebel formula then states that
\begin{equation}\label{KS}
\EP\left[M_{tu}(X)\,|\,\cF_t\right] = \frac{\mathbb{E}^\B\left[\cE^{-1}_t M_{tu}(X)\,|\,\cF_t\right]}{\mathbb{E}^\B\left[\cE^{-1}_t\,|\,\cF_t\right]}.
\end{equation}  
This can be simplified to obtain:
\begin{equation}\label{DefinitionMhat}         
\EP\left[M_{tu}(X)\,|\,\cF_t\right]=\int_{-\infty}^\infty M_{tu}(x)\,f_t(x)\rd x,
\end{equation}
where the $\cF_t$-measurable conditional density $f_t(x)$ of the random variable $X$ is given by
\begin{equation}\label{densityprocess}
f_t(x) = \frac{f_0(x)\,\exp\left(\int_0^t \ell(s,x)\rd I_s - \tfrac{1}{2}\int_0^t \ell^2(s, x)\rd s\right)}{\int_{-\infty}^\infty f_0(y)\, \exp\left(\int_0^t \ell(s,y)\rd I_s - \tfrac{1}{2}\int_0^t \ell^2(s, y)\rd s\right)\rd y}.
\end{equation}
A similar filtering system is considered in a different context in Filipovi\'c {\it et al.}~\cite{fhm}. Further conditions are imposed on the dynamics of the information process defined in (\ref{signal}) and (\ref{InfoProc}), which may be regarded necessary from a modelling point of view.

\begin{proposition}
Let $\{\cF_t\}$ be given by (\ref{Ffil}), and define the projection $\hM_{tu} =\EP\left[M_{tu}(X)\,|\,\cF_t\right]$, where $\{M_{tu}(X)\}$ is given by (\ref{martingale}). Then, for $t, u \in \re_+$, $\{\hM_{tu}\}$ is an $(\{\cF_t\}, \pr)$-martingale family.
\end{proposition}

\begin{proof}
Recall that $\cF_t \subset \cG_t$ for all $t\geq 0$. For $s\leq t$, we have
\begin{eqnarray}
\EP\left[\hM_{tu}\,|\,\cF_s\right] &=& \EP\left[\EP\left[M_{tu}(X)\,|\,\cF_t\right]\,|\,\cF_s\right],\nonumber\\
&=& \EP\left[M_{tu}(X)\,|\,\cF_s\right],\nonumber\\
&=& \EP\left[\EP\left[M_{tu}(X)\,|\,\cG_s\right]\,|\,\cF_s\right],\nonumber\\
&=& \EP\left[M_{su}(X)\,|\,\cF_s\right], \nonumber\\
&=& \hM_{su},
\end{eqnarray}
where we make use of the tower property of the conditional expectation, and the fact that $\{M_{tu}(X)\}$ is a $\{\cG_t\}$-martingale---since $\cH_t\subset\cG_t$ and $\{B_t\}$ is independent of $X$ and $\{L_t\}$.
\end{proof}
\medskip
\noindent {\bf Filtered Brownian martingales.} We consider Example \ref{BEM}, in which the total impact of uncertainties is modelled by a Brownian motion $\{W_t\}$. The corresponding filtered Esscher martingale is 
\begin{align}\label{mhat_bm}
\hM_{tu} = \int_{-\infty}^\infty f_t(x)\exp{\left(h(u,x)W_t-\tfrac{1}{2}h^2(u,x)t\right)}\rd x,
\end{align}
where the density process $\{f_t(x)\}$, given in (\ref{densityprocess}), is driven by the information process defined by (\ref{InfoProc}).

\begin{proposition}
 The filtered Brownian models have dynamics
\begin{equation}
\rd \hM_{tu} = \int_{-\infty}^\infty M_{tu}(x)f_t(x)\left[h(u,x)\rd W_t + V_t(x)\rd Z_t\right]\rd x,
\end{equation}
where 
\begin{equation}
 M_{tu}(x)=\exp\left[h(u,x)W_t-\tfrac{1}{2}h^2(u,x)t\right],
\end{equation}
\begin{equation}\label{veetx}
V_t(x) = \ell(t,x)-\EP\left[\ell(t,X)\,|\,\cF_t\right], 
\end{equation}
\begin{equation}
Z_t = I_t - \int_0^t \EP\left[\ell(s,X)\,|\,\cF_s\right]\,\rd s,
\end{equation}
and $f_t(x)$ is defined in (\ref{densityprocess}).
\end{proposition}

\begin{proof}
We first show that
\begin{equation}
\rd M_{tu}(x) = h(u,x)M_{tu}(x)\rd W_t.
\end{equation}
In Filipovi\'c {\it et al}.~\cite{fhm} it is proved that
\begin{equation}
\rd f_t(x) = f_t(x)\left(\ell(t,x)-\EP\left[\ell(t,X)\,|\,\cF_t\right]\right)\rd Z_t,
\end{equation}
where $\{Z_t\}_{t\geq 0}$ is an $(\{\cF_t\}, \pr)$-Brownian motion, defined by
\begin{equation}
Z_t = I_t - \int_0^t \EP\left[\ell(s,X)\,|\,\cF_s\right]\,\rd s.
\end{equation}
Thus by the Ito product rule, we get 
\begin{equation}
\rd[M_{tu}(x)f_t(x)] = f_t(x)\rd M_{tu}(x) + M_{tu}(x)\rd f_t(x)
\end{equation}
since $\rd W_t\,\rd Z_t = 0$. This simplifies to
\begin{equation}
\rd[M_{tu}(x)f_t(x)] = M_{tu}(x)f_t(x)\left[h(u,x)\rd W_t + \left(\ell(t,x)-\int_{-\infty}^\infty \ell(t,y)f_t(y)\rd y\right)\rd Z_t\right],
\end{equation}
and we obtain
\begin{equation}
\rd \hM_{tu} = \int_{-\infty}^\infty M_{tu}(x)f_t(x)\left[h(u,x)\rd W_t + V_t(x)\rd Z_t\right]\rd x
\end{equation}
where we define
\begin{equation}
V_t(x) = \ell(t,x) - \int_{-\infty}^\infty \ell(t,y)f_t(y)\,\rd y.
\end{equation}
\end{proof}
\begin{remark}
The dynamics of $\{\hM_{tu}\}$ can be written in the following form:
\begin{equation}\label{dyn_mhat}
\rd \hM_{tu} = \EP\left[M_{tu}(X)h(u,X)\,|\,\cF_t\right]\rd W_t + \EP\left[M_{tu}(X)V_t(X)\,|\,\cF_t\right]\rd Z_t.
\end{equation} 
\end{remark}
\medskip
\noindent {\bf Filtered gamma martingales.} Let us suppose that the total impact of uncertainties on an economy is modelled by a gamma process $\{\gamma_t\}$ with density 
\begin{equation}
\mathbb{P}(\gamma_t \in \rd y) = \frac{y^{mt-1}\exp{\left(-\frac{y}{\kappa}\right)}}{\kappa^{mt}\,\Gamma[mt]}\,\rd y,
\end{equation}
where $m$ and $\kappa$ are the rate and the scale parameter, respectively. The associated randomised Esscher martingale is given in Example \ref{GEM}, where $h(u,X)<\kappa^{-1}$. The corresponding filtered process takes the form
\begin{equation}\label{gammaMhat}
 \hM_{tu}=\int_{-\infty}^\infty f_t(x)\left(\left[1-\kappa\,h(u,x)\right]^{mt}\exp{\left[h(u,x)\,\gamma_t\right]}\right)\rd x
\end{equation}
for $h(u,x)<\kappa^{-1}$, and where the density $f_t(x)$ is given by (\ref{densityprocess}).
\\

\noindent {\bf Filtered compound Poisson and gamma martingales.} We now construct a model based on two independent L\'evy processes: a gamma process (as defined previously) and a compound Poisson process. The idea here is to use the infinite activity gamma process to represent small frequently-occurring jumps, and to use the compound Poisson process to model jumps, which are potentially much larger in magnitude, and may occur sporadically. Let $\{C_t\}_{t\geq 0}$ denote a compound Poisson process given by
\begin{equation}
C_t = \sum_{i=1}^{N_t} Y_i
\end{equation}
where $\{N_t\}_{t\geq 0}$ is a Poisson process with rate $\lambda$. The independent and identically distributed random variables $Y_i$ are independent of $\{N_t\}$. The moment generating function is given by
\begin{equation}
\EP\left[\exp{\left(\varrho\,C_t\right)}\right] = \exp{\left[\lambda t\left(M_Y(\varrho)-1\right)\right]} 
\end{equation}
where $M_Y$ is the moment generating function of $Y_i$. For $h_1(u, X) < \kappa^{-1}$, we have
\begin{eqnarray}
M_{tu}(X) &=& \frac{\exp{\left(h_1(u,X)\gamma_t + h_2(u,X)C_t\right)}}{\EP\left[\exp{\left(h_1(u,X)\gamma_t + h_2(u,X)C_t\right)}\,|\,X\right]}\nonumber\\\nonumber\\
&=& \frac{\exp{\left(h_1(u,X)\gamma_t\right)}}{\EP\left[\exp{\left(h_1(u,X)\gamma_t\right)}\,|\,X\right]}\cdot \frac{\exp{\left(h_2(u,X)C_t\right)}}{\EP\left[\exp{\left(h_2(u,X)C_t\right)}\,|\,X\right]}\nonumber\\\nonumber\\
&=& M^{(\gamma)}_{tu}(X)\;M^{(C)}_{tu}(X),
\end{eqnarray}
where, conditional on $X$, $\exp{\left(h_1(u,X)\gamma_t\right)}$ and $\exp{\left(h_2(u,X)C_t\right)}$ are independent. Furthermore,
\begin{eqnarray}
M^{(\gamma)}_{tu}(X) &=& \left(1-\kappa\,h_1(u,X)\right)^{mt}\exp{\left(h_1(u,X)\gamma_t\right)},\\
M^{(C)}_{tu}(X) &=& \exp{\left[h_2(u,X)C_t-\lambda t\left(M_Y(h_2(u,X))-1\right)\right]}.
\end{eqnarray}
Then, the filtered process takes the form
\begin{align}
 \hM_{tu} = \int_{-\infty}^\infty f_t(x)\,&\left[1-\kappa\,h_1(u,x)\right]^{mt}\nn\\
&\times \exp{\left[h_1(u,x)\,\gamma_t + h_2(u,X)C_t-\lambda t\left(M_Y(h_2(u,X))-1\right)\right]}
\rd x,
\end{align}
where $f_t(x)$ is given by (\ref{densityprocess}).


\section{Filtered Esscher martingales with L\'evy information}
\noindent Up to this point, we have considered a Brownian information process given by equation (\ref{InfoProc}). However, the noise component in the information process may be modelled by a L\'evy process with randomly sized jumps, that is independent of the L\'evy process $\{L_t\}$ used to construct the randomised Esscher martingale. In what follows, we give an example of continuously-observed information, which is distorted by gamma-distributed pure noise.
\begin{example}
{\rm
Let $\{\widetilde{\gamma}_t\}_{t\geq 0}$ be a gamma process  with rate and scale parameters $\widetilde{m}$ and $\widetilde{\kappa}$, respectively. We define the gamma information process by
\begin{equation}\label{gammaI}
I_t = X \widetilde{\gamma}_t.
\end{equation}
Brody \& Friedman \cite{bf} consider such an observation process in a similar situation. We define the filtration $\{\mathcal{G}_t\}$ by
\begin{equation}
\mathcal{G}_t = \sigma\left(\{\widetilde{\gamma}_s\}_{0\leq s\leq t}, \{L_s\}_{0\leq s\leq t}, X\right),
\end{equation}
and $\{\cF_t\}$ by
\begin{equation}
\mathcal{F}_t = \sigma\left(\{L_s\}_{0\leq s\leq t}, \{I_s\}_{0\leq s\leq t}\right)
\end{equation}
where $\{I_t\}$ is given by (\ref{gammaI}). To derive the conditional density of $X$ given $\mathcal{F}_t$, we first show that $\{I_t\}$ is a Markov process with respect to its own filtration. That is, for $a\in\mathbb{R}$,
\begin{eqnarray}
\mathbb{P}\left[I_t < a\,|\, I_s, I_{s_1}, \ldots, I_{s_n}\right] = \mathbb{P}\left[I_t < a\,|\, I_s\right]
\end{eqnarray}
for all $t\geq s\geq s_1 \geq \ldots \geq s_n\geq 0$ and for all $n\geq 1$. It follows that
\begin{eqnarray}
\mathbb{P}\left[I_t < a\,|\, I_s, I_{s_1}, \ldots, I_{s_n}\right] &=& \mathbb{P}\left[I_t < a\,\bigg|\, I_s, \frac{I_{s_1}}{I_s}, \ldots, \frac{I_{s_n}}{I_{s_{n-1}}}\right]\nn\\
&=& \mathbb{P}\left[X\,\widetilde{\gamma}_t < a\,\bigg|\, X\,\widetilde{\gamma}_s, \frac{\widetilde{\gamma}_{s_1}}{\widetilde{\gamma}_s}, \ldots, \frac{\widetilde{\gamma}_{s_n}}{\widetilde{\gamma}_{s_{n-1}}}\right].
\end{eqnarray}
It can be proven that $\widetilde{\gamma}_{s_1}/\widetilde{\gamma}_s, \ldots, \widetilde{\gamma}_{s_n}/\widetilde{\gamma}_{s_{n-1}}$ are independent of $\widetilde{\gamma}_s$ and $\widetilde{\gamma}_t$ (see Brody {\it et al}.~\cite{bhm3}). Furthermore, $\widetilde{\gamma}_{s_1}/\widetilde{\gamma}_s, \ldots, \widetilde{\gamma}_{s_n}/\widetilde{\gamma}_{s_{n-1}}$ are independent of $X$. Thus we have
\begin{eqnarray}
\mathbb{P}\left[I_t < a\,|\, I_s, I_{s_1}, \ldots, I_{s_n}\right] = \mathbb{P}\left[I_t < a\,|\, I_s\right].
\end{eqnarray}
We assume that the random variable $X$ has a continuous \textit{a priori} density $f_0(x)$. Then the conditional density of $X$,
\begin{equation}
f_t(x) = \frac{\rd}{\rd x}\mathbb{P}\left[X \leq x\,|\,I_t\right],	
\end{equation}
is given by
\begin{eqnarray}
f_t(x) &=& \frac{f_0(x)\,p\left(I_t\,|\,X=x\right)}{\int_{-\infty}^\infty f_0(y)\,p\left(I_t\,|\,X=y\right)\,\rd y}\nn\\
&=& \frac{f_0(x)x^{-\widetilde{m}t}\exp{\left[-I_t/(\widetilde{\kappa} x)\right]}}{\int_{-\infty}^\infty f_0(y)y^{-\widetilde{m}t}\exp{\left[-I_t/(\widetilde{\kappa} y)\right]}\rd y},\label{densitygamma}
\end{eqnarray}
where we have used the Bayes formula. The filtered Esscher martingale is thus obtained by 
\begin{equation}
 \widehat{M}_{tu}=\mathbb{E}\left[M_{tu}(X)\,\vert\,\mathcal{F}_t\right].
\end{equation}
The result is:
\begin{equation}
\hM_{tu} = \int_{-\infty}^\infty M_{tu}(x)\; \frac{f_0(x)x^{-\widetilde{m}t}\exp{\left[-I_t/(\widetilde{\kappa} x)\right]}}{\int_{-\infty}^\infty f_0(y)y^{-\widetilde{m}t}\exp{\left[-I_t/(\widetilde{\kappa} y)\right]}\rd y}\; \rd x.
\end{equation}
}
\end{example}

\section{Pricing kernel models}
The absence of arbitrage in a financial market is ensured by the existence of a pricing kernel $\{\pi_t\}_{t\geq 0}$ satisfying $\pi_t > 0$ almost surely for all $t\geq 0$. We consider, in general, an incomplete market and let $\{S_t\}_{t \geq 0}$ denote the price process of a non-dividend paying asset. The price of such an asset at time $t\leq T$ is given by the following pricing formula:
\begin{equation} \label{mg_pk}
S_t  = \frac{1}{\pi_t}\,\EP\left[\pi_T S_T \,|\,\cF_t\right].
\end{equation}
The price of a discount bond system with price process $\{P_{tT}\}_{0\leq t\leq T<\infty}$ and payoff $P_{TT} = 1$ is given by 
\begin{equation}\label{bondprice}
P_{tT} = \frac{1}{\pi_t}\,\EP\left[\pi_T\,|\,\cF_t\right].
\end{equation}
The specification of a model for the pricing kernel is equivalent to choosing a model for the discount bond system, and thus also for the term structure of interest rates, and the excess rate of return. A sufficient condition for positive interest rates is that $\{\pi_t\}$ be an $(\{\cF_t\}, \pr)$-supermartingale. If, in addition, the value of a discount bond should vanish in the limit of infinite maturity, then $\{\pi_t\}$ must satisfy
\begin{equation} 
\lim_{T\rightarrow \infty}\EP\left[\pi_T\right] = 0.
\end{equation}
A positive right-continuous supermartingale with this property is called a potential. Let $\{A_t\}_{t\geq 0}$ be an $\{\cF_t\}$-adapted process with right-continuous non-decreasing paths, where $A_0 = 0$ almost surely, and let $\{A_t\}$ be integrable, that is, $\mathbb{E}\left[A_\infty\right] < \infty$ where
$A_\infty := \lim_{t\rightarrow\infty}{A_t}$. Then any right-continuous version of the supermartingale
\begin{equation}\label{classDpot}
\zeta_t = \mathbb{E}\left[A_\infty\, |\, \cF_t\right]- A_t
\end{equation}
is a potential of class (D)\footnote{A right-continuous $\{\cF_t\}$-adapted stochastic process $\{X_t\}_{t\geq 0}$ is said to belong to the class (D) if the random variables $X_\tau$ are uniformly integrable, where $\tau$ is any finite $\{\cF_t\}$-stopping time.}, see Meyer \cite{mey}. Let us denote by $\{\zeta_t\}_{t\geq 0}$ the potential generated by $\{A_t\}$. Meyer \cite{mey} proved that a potential belongs to the class (D) if, and only if, it is generated by a process $\{A_t\}$. Thus, it is enough to choose a
process $\{A_t\}$ to model the pricing kernel.

Flesaker \& Hughston \cite{fh} provide a framework for constructing positive interest rate models, in which the pricing kernel is modelled by
\begin{equation}\label{FHpk}
\pi_t = \int_t^\infty \rho(u)\,m_{tu}\,\rd u,
\end{equation}
where $\{m_{tu}\}_{0\leq t\leq u<\infty}$ is a family of positive unit-initialized martingales, and 
\begin{equation}
\rho(t) = -\partial_t P_{0t}.
\end{equation}
It can be shown that the pricing kernel (\ref{FHpk}) is a potential generated by
\begin{equation}
A_t = \int_0^t \rho(u)\, m_{uu}\,\rd u,
\end{equation}
and thus, that it is a potential of class (D). Furthermore, given a potential (\ref{classDpot}) where $\{A_t\}$ is an increasing, integrable process of the form
\begin{equation}\label{abs}
A_t = \int_0^t a_s\, \rd s,
\end{equation}
with $\{a_t\}_{t\geq 0}$ a nonnegative process, there exist a deterministic function
\begin{equation}
\rho(u) = \frac{1}{\pi_0}\,\EP\left[a_u\right],
\end{equation}
and a positive martingale 
\begin{equation}
m_{tu} = \frac{\EP\left[a_u \,|\,\cF_t\right]}{\EP\left[a_u\right]}
\end{equation}
for each fixed $u \geq t$ where $m_{0u}=1$, such that the class (D) potential can be written in the form 
\begin{equation}
\pi_t = \pi_0\,\int_t^\infty \rho(u)\,m_{tu}\,\rd u,
\end{equation}
where $\pi_0$ is a scaling factor. Thus, the Flesaker-Hughston models are precisely the class of pricing kernels that are class (D) potentials where $\{A_t\}$ is increasing, integrable and of the form (\ref{abs}), see Hunt \& Kennedy \cite{hk}. Therefore, to model such class (D) potentials, it suffices to specify a family of positive martingales. 

In what follows, we construct explicit Flesaker-Hughston models, which are driven by a randomised mixture of L\'evy processes. We develop such a class of pricing kernels by setting
\begin{equation}\label{fhpk}
\pi_t = \int_t^\infty \rho(u)\,\hM_{tu}\,\rd u
\end{equation}
where the martingale family $\{\hM_{tu}\}_{0\leq t\leq u < \infty}$ is defined by (\ref{hatM}) with $\hM_{tu}>0$ and $\hM_{0u} =1$. Then, the discount bond system is given by
\begin{equation}\label{fhbp}
P_{tT} = \frac{\int_T^\infty \rho(u)\,\hM_{tu}\,\rd u}{\int_t^\infty \rho(u)\,\hM_{tu}\,\rd u}.
\end{equation}
The associated instantaneous forward rate $\{r_{tT}\}_{0\leq t\leq T}$ is defined by $r_{tT} = -\partial_T \ln{P_{tT}}$. We deduce that
\begin{equation}\label{forward}
r_{tT} = \frac{\rho(T)\,\hM_{tT}}{\int_T^\infty \rho(u)\,\hM_{tu}\,\rd u},
\end{equation}
and that the short rate of interest $\{r_t\}_{t\geq 0}$ is given by the formula
\begin{equation}\label{fhshort}
r_t = \frac{\rho(t)\,\hM_{tt}}{\int_t^\infty \rho(u)\,\hM_{tu}\,\rd u},
\end{equation}
where $r_t := r_{tt}$. The interest rate is positive by construction. We note here that the pricing kernel models proposed in Brody {\it et al.}~\cite{bhmack} can be recovered by considering a special case of the random mixer, namely $h(u,X) = h(u)$.


\section{Pricing kernel models driven by filtered Brownian martingales}
In the case where the filtered martingales driving the pricing kernel are Gaussian processes, the dynamics of the discount bond system can be expressed by a diffusion equation of the form (\ref{P-diffusioneqn}). Inserting the filtered Brownian martingale family (\ref{mhat_bm}) into (\ref{fhbp}), we obtain the price process of the discount bond in the Brownian set-up:

\begin{equation}
P_{tT} = \frac{\int_T^\infty \rho(u)\, \int_{-\infty}^\infty f_t(x)\,\exp{\left[h(u,x)W_t - \tfrac{1}{2}h^2(u,x)t\right]}\,\rd x\, \rd u}{\int_t^\infty \rho(v)\, \int_{-\infty}^\infty f_t(y)\,\exp{\left[h(v,y)W_t - \tfrac{1}{2}h^2(v,y)t\right]}\,\rd y\, \rd v}.
\end{equation}
A similar expression is obtained for the associated interest rate system by plugging (\ref{mhat_bm}) into (\ref{fhshort}).

\begin{proposition}
The dynamical equation of the discount bond process is given by 
\begin{equation}\label{P-diffusioneqn}
\frac{\rd P_{tT}}{P_{tT}} = \left[r_t -\theta_{tt}(\theta_{tT}-\theta_{tt}) - \nu_{tt}(\nu_{tT}-\nu_{tt}) \right]\rd t + (\theta_{tT}-\theta_{tt})\rd W_t + (\nu_{tT}-\nu_{tt})\rd Z_t
\end{equation}
where
\begin{eqnarray}
\theta_{tT} &:=& \frac{\int_T^\infty \rho(u)\,\EP\left[M_{tu}(X)h(u,X)\,|\,\cF_t\right]\,\rd u}{\int_T^\infty \rho(u)\,\hM_{tu}\,\rd u},\\
\nu_{tT} &:=& \frac{\int_T^\infty \rho(u)\,\EP\left[M_{tu}(X)V_t(X)\,|\,\cF_t\right]\,\rd u}{\int_T^\infty \rho(u)\,\hM_{tu}\,\rd u},
\end{eqnarray}
$\theta_{tt} = \theta_{tT}\big|_{T=t}$, and $\nu_{tt} = \nu_{tT}\big|_{T=t}$. 
\end{proposition}
\begin{proof}
First we have
\begin{equation}
\rd \left[\int_T^\infty \rho(u)\,\hM_{tu}\,\rd u\right] =  \int_T^\infty \rho(u)\,\rd \hM_{tu}\,\rd u
\end{equation}
where $\rd \hM_{tu}$ is given by (\ref{dyn_mhat}). Also,
\begin{equation}
\rd \left[\int_t^\infty \rho(u)\,\hM_{tu}\,\rd u\right] =  \int_t^\infty \rho(u)\,\rd \hM_{tu}\,\rd u -\rho(t)\,\hM_{tt}\,\rd t.
\end{equation}
We then apply the Ito quotient rule to obtain the dynamics of $\{P_{tT}\}$. We observe that the discount bond volatilities are given by
\begin{eqnarray}
\Omega^{(1)}_{tT} &=& \theta_{tT}-\theta_{tt},\\
\Omega^{(2)}_{tT} &=& \nu_{tT}-\nu_{tt}.
\end{eqnarray}
The market price of risk associated with $\{W_t\}$ is $\lambda^{(1)}_t := -\theta_{tt}$; the one associated with $\{Z_t\}$ is $\lambda^{(2)}_t := -\nu_{tt}$. The product between the bond volatility vector $\Omega_{tT} = (\Omega^{(1)}_{tT}, \Omega^{(2)}_{tT})$ and the market price of risk vector $\lambda_t = (\lambda^{(1)}_t, \lambda^{(2)}_t)$  gives us the risk premium associated with an investment in the discount bond, that is,
\begin{equation}
\Omega_{tT} \cdot \lambda_t = -\theta_{tt}\left(\theta_{tT}-\theta_{tt}\right) - \nu_{tt}\left(\nu_{tT}-\nu_{tt}\right).
\end{equation} 
\end{proof}

\begin{proposition}
Let $\{M_{tu}(X)\}$ be of the class (\ref{examplebm}), and let $\{\hM_{tu}\}$ in (\ref{forward}) be given by the martingale family (\ref{mhat_bm}). Then the dynamical equation of the forward rate is given by
\begin{equation}
\rd r_{tT} = \left[\theta_{tT}\,\partial_T \theta_{tT} + \nu_{tT}\,\partial_T \nu_{tT}\right]\rd t - \partial_T \theta_{tT}\,\rd W_t - \partial_T \nu_{tT}\,\rd Z_t
\end{equation}
where
\begin{equation}
\theta_{tT} := \frac{\int_T^\infty \rho(u)\,\EP\left[M_{tu}(X)h(u,X)\,|\,\cF_t\right]\,\rd u}{\int_T^\infty \rho(u)\,\hM_{tu}\,\rd u},
\end{equation}
and
\begin{equation}
\nu_{tT} := \frac{\int_T^\infty \rho(u)\,\EP\left[M_{tu}(X)V_t(X)\,|\,\cF_t\right]\,\rd u}{\int_T^\infty \rho(u)\,\hM_{tu}\,\rd u}
\end{equation}
where $V_t(X)$ is defined by (\ref{veetx}).
\end{proposition}

\begin{proof}
We apply the Ito quotient rule to (\ref{forward}) to obtain the forward rate dynamics. We make the observations that 
\begin{equation}
\partial_T \theta_{tT} = r_{tT}\left(\theta_{tT}-\frac{\EP\left[M_{tT}(X)h(T,X)\,|\,\cF_t\right]}{\hM_{tT}}\right),
\end{equation}
and that
\begin{equation}
\partial_T \nu_{tT} = r_{tT}\left(\nu_{tT}-\frac{\EP\left[M_{tT}(X)V_t(X)\,|\,\cF_t\right]}{\hM_{tT}}\right).
\end{equation} 
\end{proof}
\noindent In particular, if we set
\begin{eqnarray}
\Sigma_{tT} &=& \theta_{tT} - \theta_{tt},\\
\Lambda_{tT} &=& \nu_{tT} - \nu_{tt},
\end{eqnarray}
then we can express the risk-neutral dynamics of the forward rate by
\begin{equation}\label{hjmforwarddynamics}
\rd r_{tT} = \left[\Sigma_{tT}\partial_T \Sigma_{tT}+\Lambda_{tT}\partial_T \Lambda_{tT}\right]\rd t - \partial_T \Sigma_{tT}\rd \widetilde{W}_t - \partial_T \Lambda_{tT}\rd \widetilde{Z}_t,
\end{equation}
where $\{\widetilde{W}_t\}_{t\geq 0}$ and  $\{\widetilde{Z}_t\}_{t\geq 0}$ are Brownian motions defined by the Girsanov relations
\begin{eqnarray}
\rd \widetilde{W}_t &=& \rd W_t + \lambda^{(1)}_t\,\rd t,\nonumber\\
\rd \widetilde{Z}_t &=& \rd Z_t +\lambda^{(2)}_t\,\rd t.
\end{eqnarray} 
The dynamical equation (\ref{hjmforwarddynamics}) has the form of the HJM dynamics for the forward rate under the risk-neutral measure, see Heath {\it et al}.~\cite{hjm}. 

\begin{example}\label{BrownianMotionExampleOne}
{\rm
As a first illustration, let us now consider the case in which the information process is defined by 
\begin{equation}\label{Ibhm}
I_t = \sigma X t + B_t,
\end{equation}
where $\sigma$ is a positive constant. It can be proven that this is a Markov process (see Brody {\it et al}.~\cite{bhm}). Equation (\ref{Ibhm}) is a special case of the path-dependent observation process (\ref{Inonmarkov}). Let $\{W_t\}$ be a standard Brownian motion that is independent of $X$. Then from Example \ref{BEM}, we have
\begin{equation}
M_{tu}(X) = \exp{\left[h(u,X)W_t - \tfrac{1}{2}\,h^2(u,X)t\right]}.
\end{equation}
We suppose that the \textit{a priori} distribution of $X$ is uniform over the interval $(a,b)$, where $a\geq 0$ and $b>0$. We choose to model the random mixer by
\begin{equation}
h(u,X) = c\exp{\left(-u X\right)}
\end{equation}
where $c\in \re$. Here $X$ can be interpreted as the random rate of the exponential decay in $h(u,X)$. We obtain the following expressions for the bond price
\begin{equation}
P_{tT} = \frac{\int_T^\infty \rho(u)\, \int_a^{b} \exp{\left[\sigma x I_t + c\textup{e}^{-u x}W_t - \tfrac{1}{2}\left(\sigma^2 x^2 + c^2\textup{e}^{-2u x} \right)t\right]} \,\rd x\,\rd u}{\int_t^\infty \rho(u)\, \int_a^{b} \exp{\left[\sigma y I_t + c\textup{e}^{-u y}W_t - \tfrac{1}{2}\left(\sigma^2 y^2 + c^2\textup{e}^{-2u y} \right)t\right]} \,\rd y\,\rd u},
\end{equation}
and the associated interest rate
\begin{equation}
r_t = \frac{\rho(t)\,\int_a^{b} \exp{\left[\sigma x I_t + c\textup{e}^{-t x}W_t - \tfrac{1}{2}\left(\sigma^2 x^2 + c^2\textup{e}^{-2t x} \right)t\right]} \,\rd x}{\int_t^\infty \rho(u)\, \int_a^{b} \exp{\left[\sigma y I_t + c\textup{e}^{-u y}W_t - \tfrac{1}{2}\left(\sigma^2 y^2 + c^2\textup{e}^{-2u y} \right)t\right]} \,\rd y}.
\end{equation}
Since the model is constructed from a single L\'evy process, it is not --- strictly speaking --- a mixture model as described previously. However, it can be viewed as a kind of two-factor Brownian model owing to the presence of the observation process $\{I_t\}$. The bond price and the associated interest rate are functions of time and the two state variables $W_t$ and $I_t$. Thus, it is straightforward to generate simulated sample paths:
 
\begin{figure}[H]
\begin{center}
\includegraphics[scale=0.80]{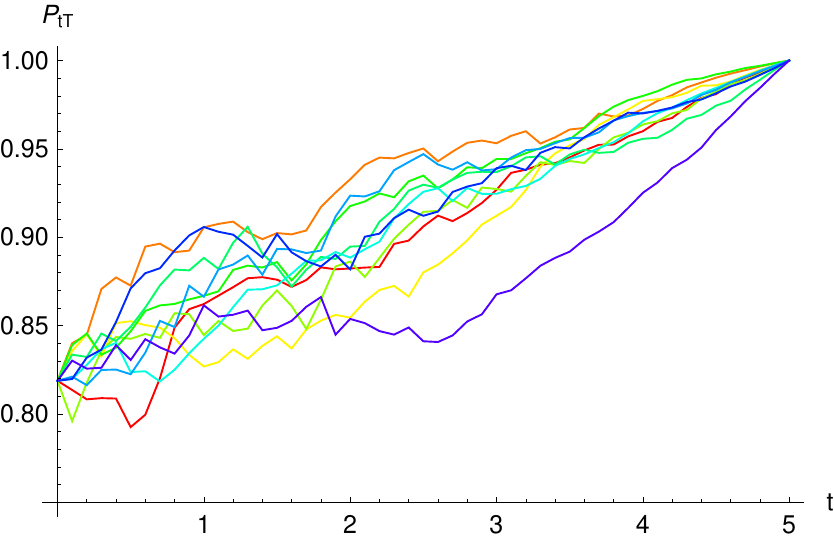}
\includegraphics[scale=0.80]{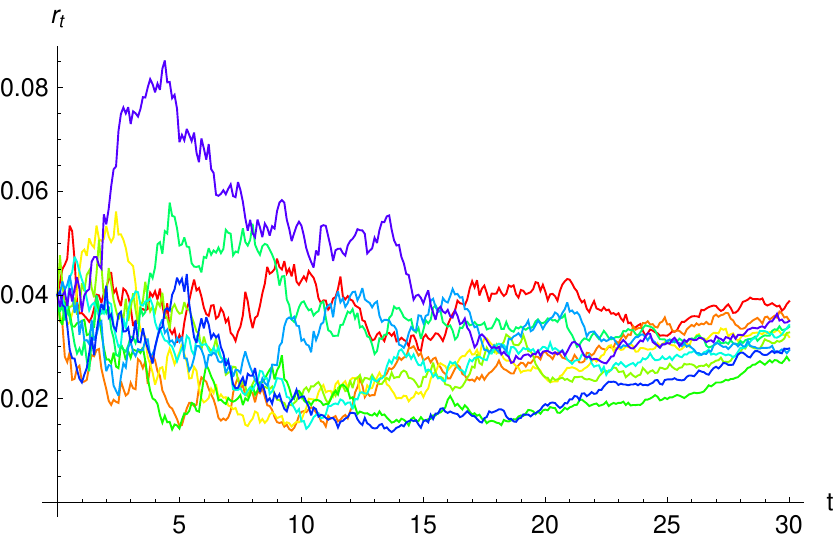}
\caption{Sample paths of discount bond with $T=5$ and short rate. We use the filtered Brownian model with $h(u,X) = c\exp{\left(-u X\right)}$ and $X\sim U(a,b)$. We set $a=0$, $b=0.1$, $\sigma=0.1$, $c=0.5$ and $P_{0t}=\exp{\left(-0.04 t\right)}$.}
\vspace{-0.5cm}
\label{figBB}
\end{center}
\end{figure}
The parameters $a$ and $b$ influence the rate at which $\exp{(-uX)}$ decays, and together with $c$ determine the impact of the Brownian motion $\{W_t\}$ on the bond and interest rate evolution. When $c$ is close to zero, the impact of $\{W_t\}$ is very small. For sufficiently large values of $b-a$, $\sigma$ or $|c|$, the numerical integration in the calculation of the pricing kernel may fail to converge. For large values of $t$, we observe that the sample paths of the short rate revert to $r_0$. Thus, there is built-in reversion to the initial level of the short rate.
}
\end{example}


\section{Bond prices driven by filtered gamma martingales}
Let $\{\gamma_t\}$ denote a gamma process with $\EP[\gamma_t] = \kappa mt$, and $\textup{Var}[\gamma_t] = \kappa^2 mt$. We consider a bond price model based on a pricing kernel that is driven by a family of filtered gamma martingales given by (\ref{gammaMhat}).
Then, equation (\ref{fhbp}) for the bond price gives the following expression:
\begin{equation}
P_{tT} = \frac{\int_T^\infty \rho(u)\, \int_{-\infty}^\infty f_t(x)\,\left[1-\kappa h(u,x)\right]^{mt}\, \exp{\left[h(u,x)\gamma_t\right]}\,\rd x\,\rd u}{\int_t^\infty \rho(v)\, \int_{-\infty}^\infty f_t(y)\,\left[1-\kappa h(v,y)\right]^{mt}\, \exp{\left[h(v,y)\gamma_t\right]}\,\rd y\,\rd v}.
\end{equation}
We now investigate this bond price model in more detail, and in particular show the effects of the various model components on the behaviour of the bond price. 

\begin{example}\label{GammaExamplewithBrownianInfo}
{\rm
Let the information process $\{I_t\}$, driving the conditional density $\{f_t(x)\}$ be of the form
\begin{equation}
I_t = \sigma t X + B_t,
\end{equation}
where $X$ is a binary random variable taking the values $X=1$ with \textit{a priori} probability $f_0(1)$, and $X=0$ with probability $f_0(0)$. We choose the random mixer
\begin{equation}
h(u,X)= c\exp{\left[-bu(1-X)\right]},
\end{equation}
where $c<\kappa^{-1}$ and $b>0$. Then the expression for the filtered gamma martingale simplifies to
\begin{equation}
\hM_{tu} = f_t(0)\exp{\left(c\textup{e}^{-bu}\gamma_t\right)}\left(1-\kappa c\textup{e}^{-bu}\right)^{mt} + f_t(1)\exp{\left(c\gamma_t\right)}\left(1-\kappa c\right)^{mt},
\end{equation}
where
\begin{align}
&f_t(0) = \frac{f_0(0)}{f_0(0) + f_0(1)\exp{\left(\sigma I_t - \tfrac{1}{2}\sigma^2 t\right)}}&
&\;\;f_t(1) = \frac{f_0(1)\exp{\left(\sigma I_t - \tfrac{1}{2}\sigma^2 t\right)}}{f_0(0) + f_0(1)\exp{\left(\sigma I_t - \tfrac{1}{2}\sigma^2 t\right)}}.&
\end{align}

There are a number of degrees of freedom in this model which have a significant impact on the behaviour of the trajectories. In what follows, we analyse the degrees of freedom one by one.
\\\\
\noindent\textbf{\textit{A priori} probability:} When $f_0(1)=0$, the diffusion $\{I_t\}$ plays no role. The sample paths of the discount bond and the short rate are driven solely by the pure jump process. The size of the jumps decays over time. As $f_0(1)$ increases, there is a greater amount of diffusion in the sample paths. Furthermore, there is a higher likelihood of obtaining sample paths for which the size of the jumps do not decay over time. If $f_0(1)=1$, then $\{\hM_{tu}\}$ is no longer $u$ dependent. This yields a stochastic pricing kernel, but flat short rate and deterministic discount bond prices, see Figure \ref{figGBp}.

\begin{figure}[H]
\begin{center}
\includegraphics[scale=0.56]{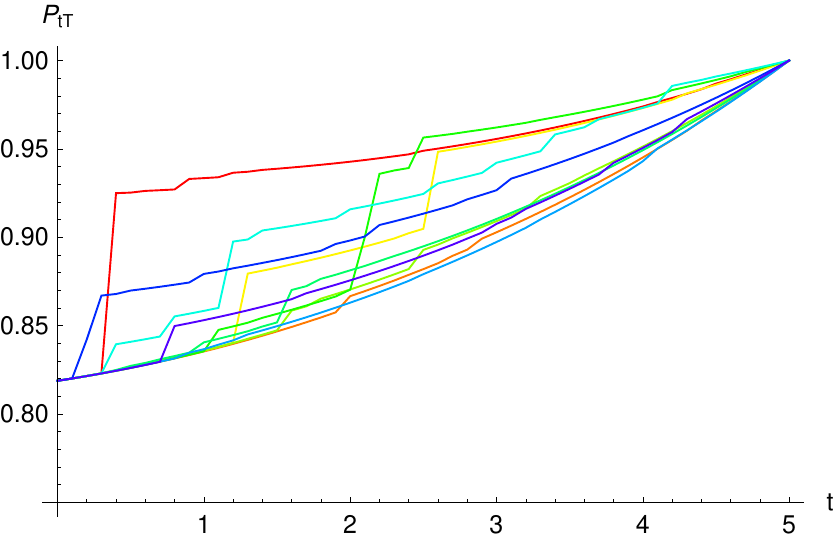}
\includegraphics[scale=0.56]{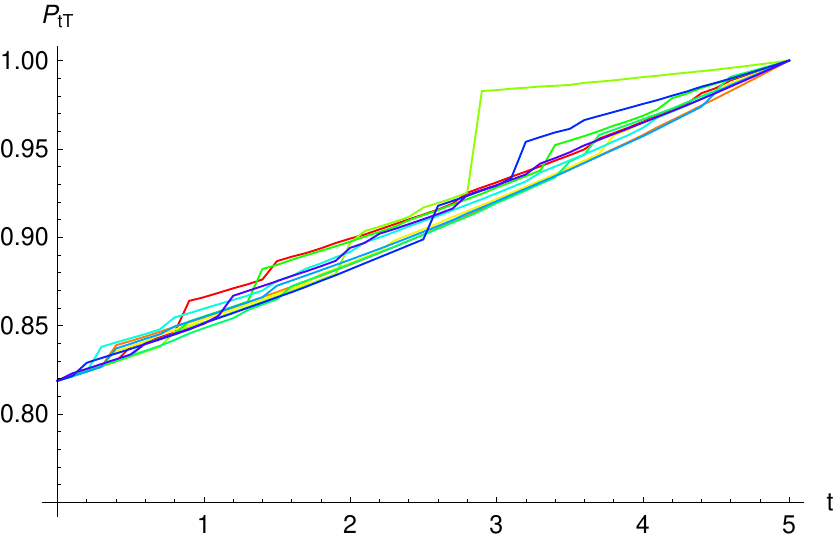}
\includegraphics[scale=0.56]{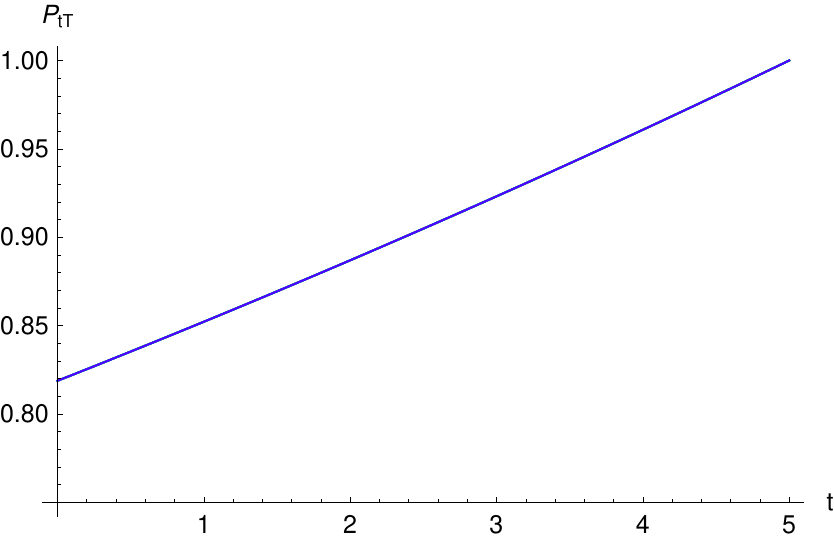}

\includegraphics[scale=0.56]{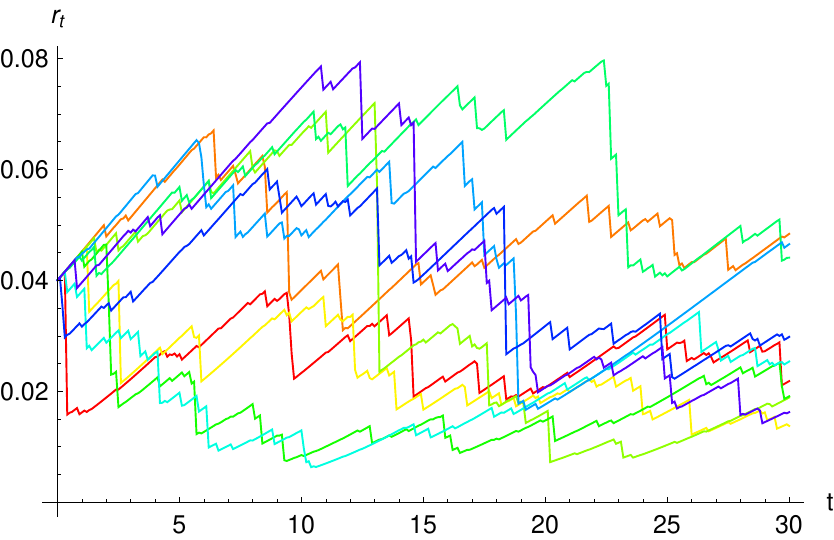}
\includegraphics[scale=0.56]{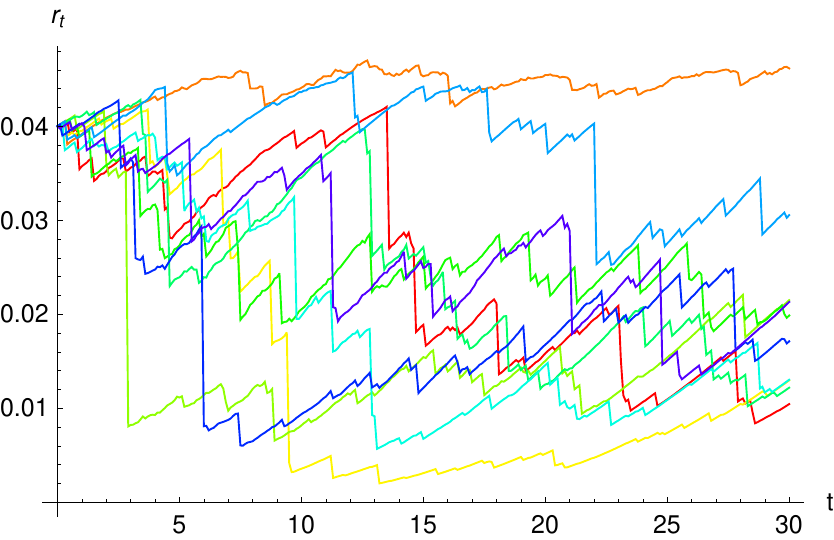}
\includegraphics[scale=0.56]{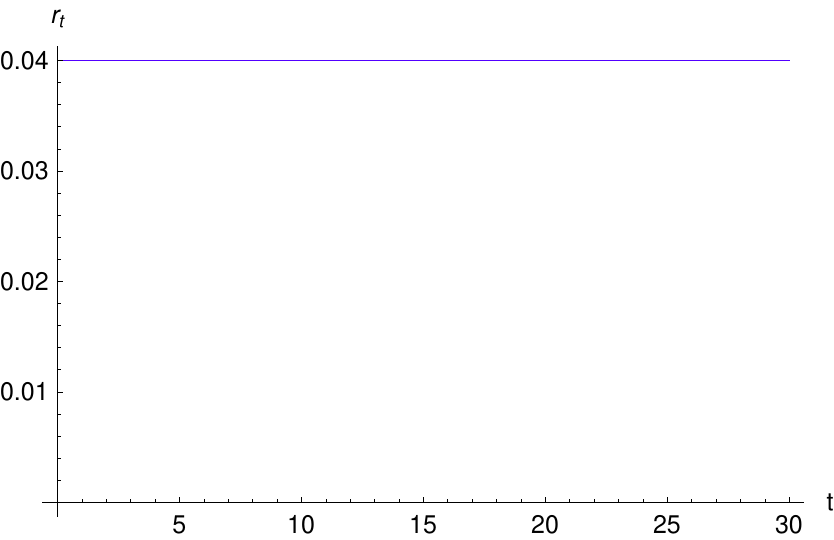}
\caption{Sample paths for discount bond with $T=5$, and associated short rate. We use the Brownian-gamma model with $h(u,X) = c\exp{[-bu(1-X)]}$ where $X=\{0,1\}$ with $m=0.5$, $\kappa=0.5$, $\sigma=0.1$, $c=-2$, $b=0.03$ and $P_{0t}=\exp{\left(-0.04 t\right)}$. We let $(i)\;f_0(1)=0$, $(ii)\; f_0(1)=0.65$ and $(iii)\;f_0(1)=1$.}\vspace{-0.5cm}
\label{figGBp}
\end{center}
\end{figure}

\noindent \textbf{Information flow rate $\sigma$:} As the information flow rate increases, the investor becomes more knowledgeable at an earlier stage about whether the random variable may take the value $X=0$ or $X=1$, see Figure \ref{figGBsigma}.
\\\\
\noindent
\textbf{Parameters of the gamma process $m$ and $\kappa$:} The rate parameter $m$ controls the rate of jump arrivals. The scale parameter $\kappa$ controls the jump size.
\\\\
\noindent
\textbf{Parameters of the random mixer $b$ and $c$:}  The magnitude of $c$ influences the impact of the jumps on the interest rate dynamics. When $c=0$, the pricing kernel, and thus the short rate of interest, is deterministic. The sign of $c$ affects the direction of the jumps. For $0<c<\kappa^{-1}$, the short rate (discount bond) sample paths have upward (downward) jumps. The opposite is true for $c<0$. It should be noted that $\exp{\left(c\exp{\left[-bu(1-X)\right]\gamma_t}\right)}$, and $\left(1-\kappa c \exp{\left[-bu(1-X)\right]}\right)^{mt}$ behave antagonistically in $c$. For large $t$, one term will eventually dominate the other. Thus, for both $c>0$ and $c<0$, the drift of the short rate trajectories is initially negative and then becomes positive for large $t$, see Figure \ref{figGBc}. The parameter $b$ determines how quickly the jumps are ``killed off''. Alternatively, $b$ can be viewed as the rate of reversion to the initial level of the interest rate. The interest rate process approaches the initial rate more rapidly for high values of $b$. When $b=0$, $\hM_{tu}$ is no longer $u$ dependent, and we obtain a stochastic pricing kernel, but flat short rate and deterministic discount bond prices, see Figure \ref{figGBb}.

\begin{figure}[H]
\begin{center}
\includegraphics[scale=0.56]{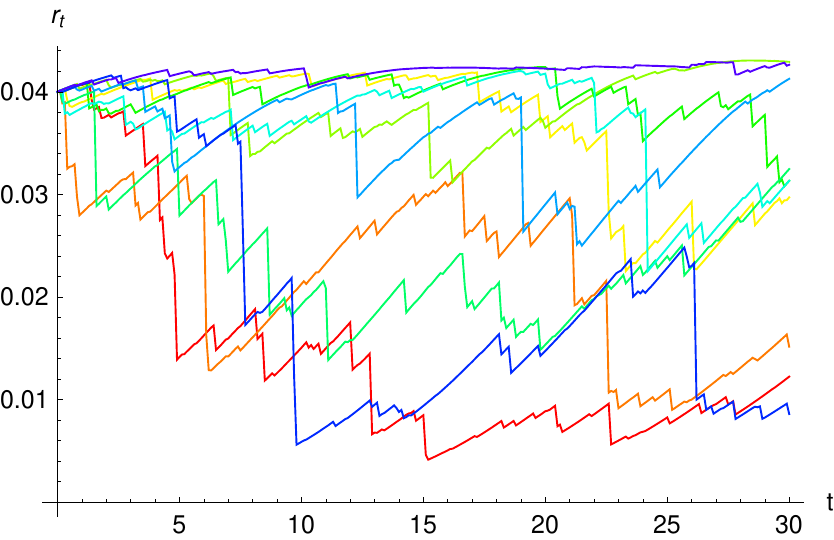}
\includegraphics[scale=0.56]{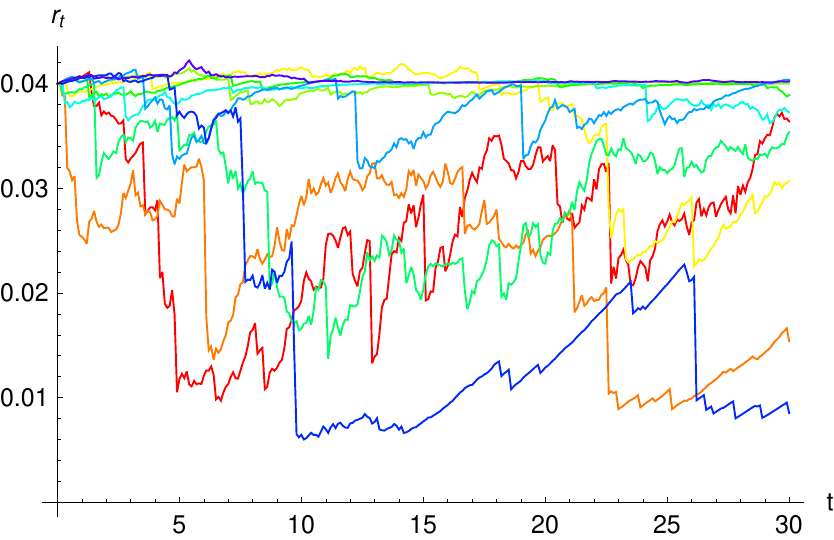}
\includegraphics[scale=0.56]{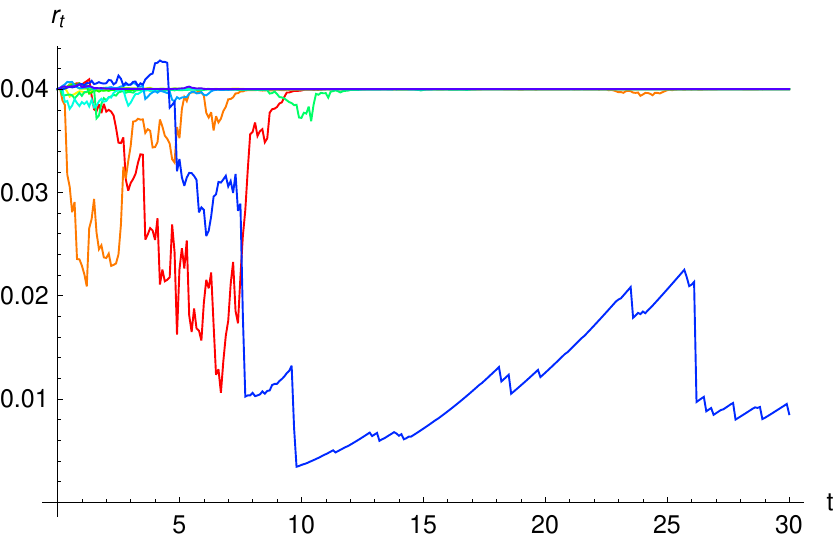}
\caption{Short rate sample paths for the Brownian-gamma model with $h(u,X) = c\exp{[-bu(1-X)]}$ and $X=\{0,1\}$. We choose $m=0.5$, $\kappa=0.5$, $f_0(1)=0.8$, $c=-2$, $b=0.03$ and $P_{0t}=\exp{\left(-0.04 t\right)}$. We set $(i)\;\sigma=0.005$, $(ii)\; \sigma=0.4$ and $(iii)\;\sigma=1.2$.}\vspace{-0.5cm}
\label{figGBsigma}
\end{center}
\end{figure}

\begin{figure}[H]
\begin{center}
\includegraphics[scale=0.56]{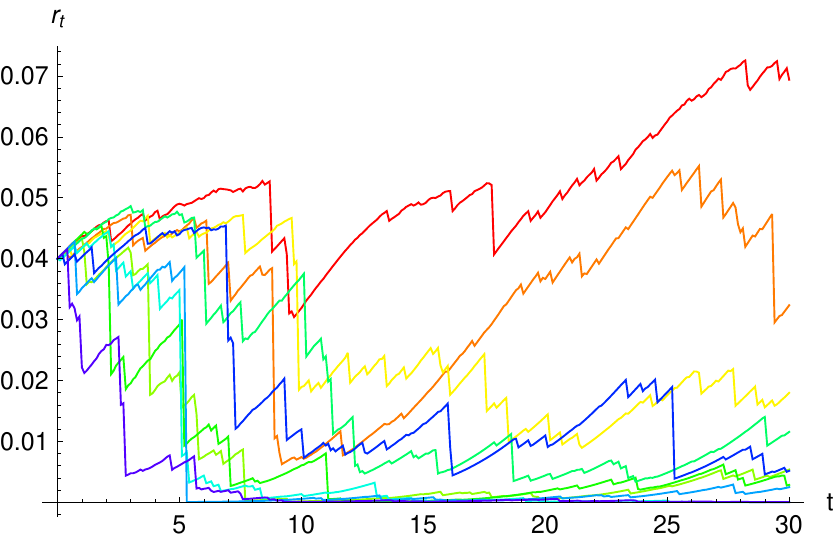}
\includegraphics[scale=0.56]{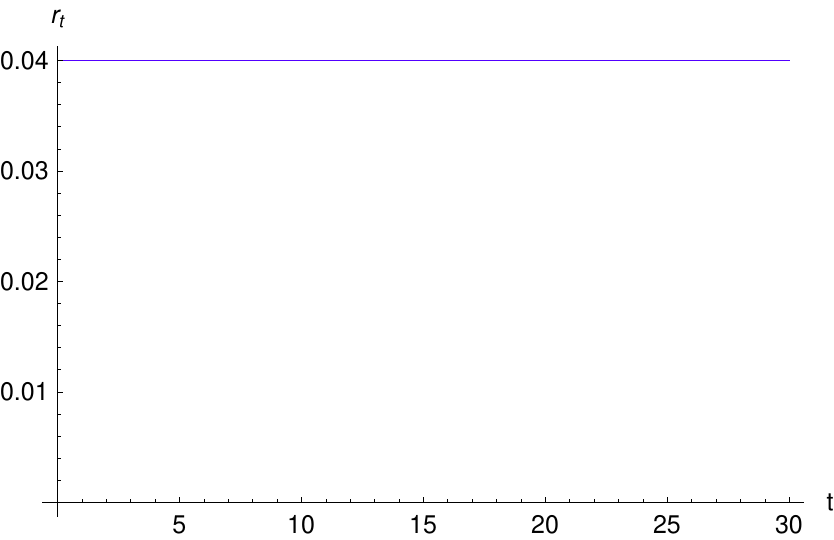}
\includegraphics[scale=0.56]{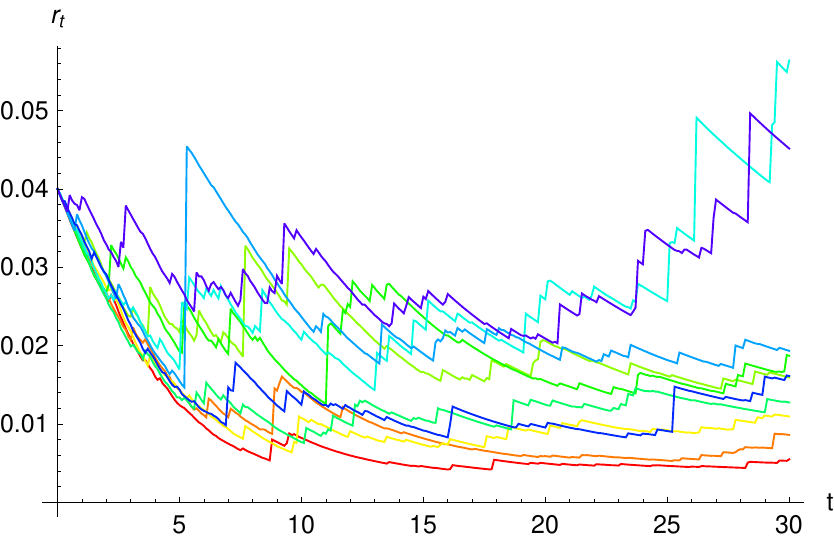}
\caption{Short rate sample paths for the Brownian-gamma model with $h(u,X) = c\exp{[-bu(1-X)]}$ and $X=\{0,1\}$. We set $m=0.5$, $\kappa=0.5$, $f_0(1)=0.5$, $\sigma=0.1$, $b=0.03$ and $P_{0t}=\exp{\left(-0.04 t\right)}$. We choose $(i)\; c = -5$, $(ii)\; c=0$ and $(iii)\;c=1.5$.}\vspace{-0.5cm}
\label{figGBc}
\end{center}
\end{figure}

\begin{figure}[H]
\begin{center}
\includegraphics[scale=0.56]{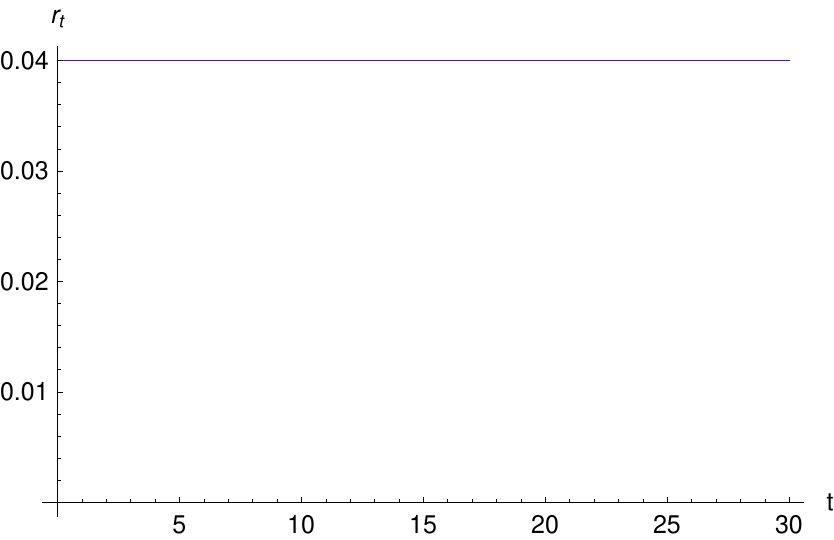}
\includegraphics[scale=0.56]{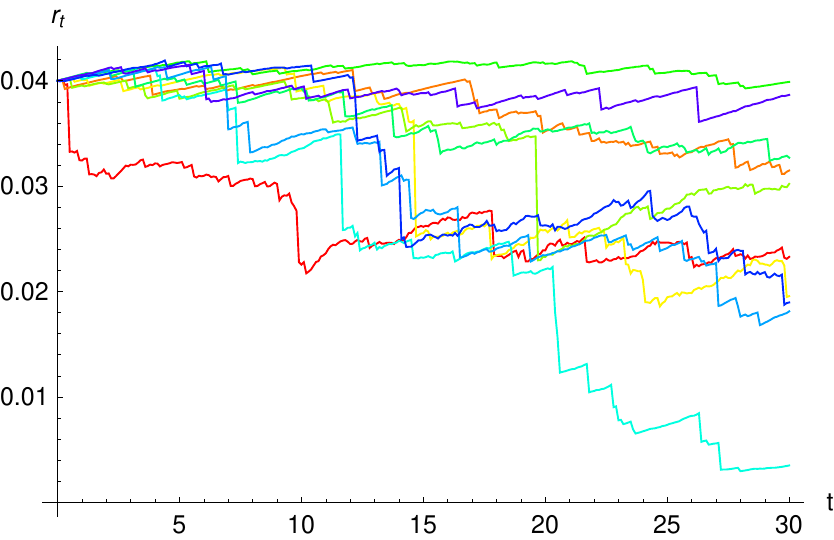}
\includegraphics[scale=0.56]{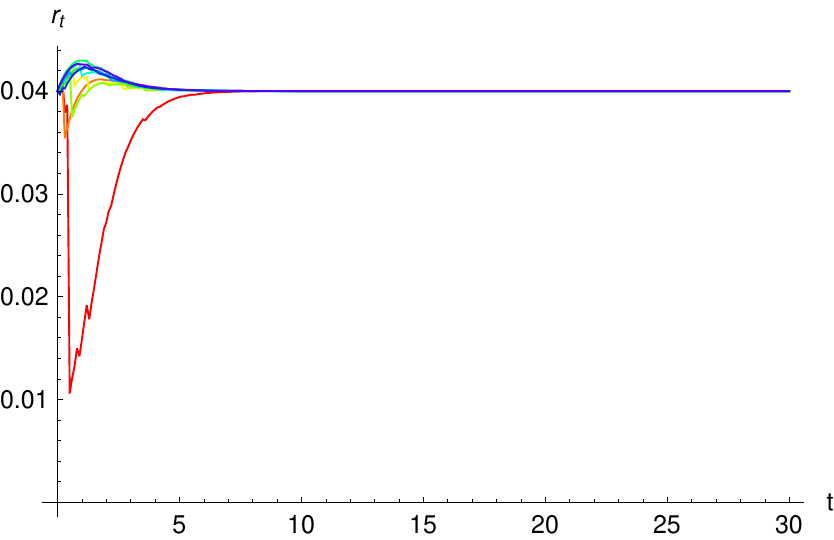}
\caption{Short rate sample paths for the Brownian-gamma model with $h(u,X) = c\exp{[-bu(1-X)]}$ and $X=\{0,1\}$. We let $m=0.5$, $\kappa=0.5$, $f_0(1)=0.5$, $\sigma=0.1$, $c=-2$ and $P_{0t}=\exp{\left(-0.04 t\right)}$. We choose $(i)\; b = 0$, $(ii)\; b=0.005$ and $(iii)\;b=1$.}\vspace{-0.5cm}
\label{figGBb}
\end{center}
\end{figure}
}
\end{example}
\noindent Compared to Example \ref{BrownianMotionExampleOne}, this model is more robust to variation in the values of the parameters. An analysis of the sample trajectories suggests that for large $t$, the short rate reverts to the initial level $r_0$.


\section{Bond prices driven by filtered variance-gamma martingales}
We let $\{L_t\}$ denote a variance-gamma process. We define the variance-gamma process as a time-changed Brownian motion with drift (see Carr {\it et al.}~\cite{ccm}), that is
\begin{equation}
L_t = \theta \gamma_t + \Sigma B_{\gamma_t}
\end{equation}
with parameters $\theta\in\mathbb{R}$, $\Sigma>0$ and $\nu>0$. Here $\{\gamma_t\}$ is a gamma process with rate and scale parameters $m=1/\nu$ and $\kappa = \nu$ respectively, and $\{B_{\gamma_t}\}$ is a subordinated Brownian motion. The randomised Esscher martingale is expressed by
\begin{equation}
M_{tu}(X) = \exp{\left[h(u,X)L_t\right]}\left(1-\theta\nu h(u,X)-\tfrac{1}{2}\Sigma^2\nu h^2(u,X)\right)^{t/\nu},
\end{equation}
and the associated filtered Esscher martingale is of the form
\begin{equation}\label{hatmvg}
\hM_{tu} = \int_{-\infty}^\infty f_t(x) \exp{\left[h(u,x)L_t\right]}\left(1-\theta\nu h(u,x)-\tfrac{1}{2}\Sigma^2\nu h^2(u,x)\right)^{t/\nu}\rd x,
\end{equation}
where $f_t(x)$ may be given for example by (\ref{densityprocess}) or a special case thereof, or by (\ref{densitygamma}) depending on the type of information used to filter knowledge about $X$. This leads to the following expression for the discount bond price process:
\begin{equation}
P_{tT} = \frac{\int_T^\infty \rho(u)\, \int_{-\infty}^\infty f_t(x)\,\exp{\left[h(u,x)L_t\right]}\left(1-\theta\nu h(u,x)-\tfrac{1}{2}\Sigma^2\nu h^2(u,x)\right)^{t/\nu}\,\rd x\,\rd u}{\int_t^\infty \rho(v)\, \int_{-\infty}^\infty f_t(y)\,\exp{\left[h(v,y)L_t\right]}\left(1-\theta\nu h(v,y)-\tfrac{1}{2}\Sigma^2\nu h^2(v,y)\right)^{t/\nu}\,\rd y\,\rd v}.
\end{equation}
We can also obtain an expression for the short rate of interest by substituting (\ref{hatmvg}) into (\ref{fhshort}). We now present another explicit bond pricing model.

\begin{example}\label{VGmod}
{\rm
We assume that $X$ is a random time, and hence a positive random variable taking discrete values $\{x_1, \ldots, x_n\}$ with \textit{a priori} probabilities $\{f_0(x_1), \ldots, f_0(x_n)\}$. We suppose that the information process $\{I_t\}$ is independent of $\{L_t\}$, and that it is defined by
\begin{equation}
I_t = \sigma X t + B_t.
\end{equation}
We take the random mixer to be
\begin{equation}
h(u,X) = c\exp{\left[-b(u-X)^2\right]}
\end{equation}
where $b>0$ and $c\in\mathbb{R}$. We see in Figure \ref{hwaves} that the random mixer, and thus the weight of the variance-gamma process, increases (in absolute value) until the random time $X$, and decreases (in absolute value) thereafter.
\begin{figure}[H]
\begin{center}
\includegraphics[scale=0.55]{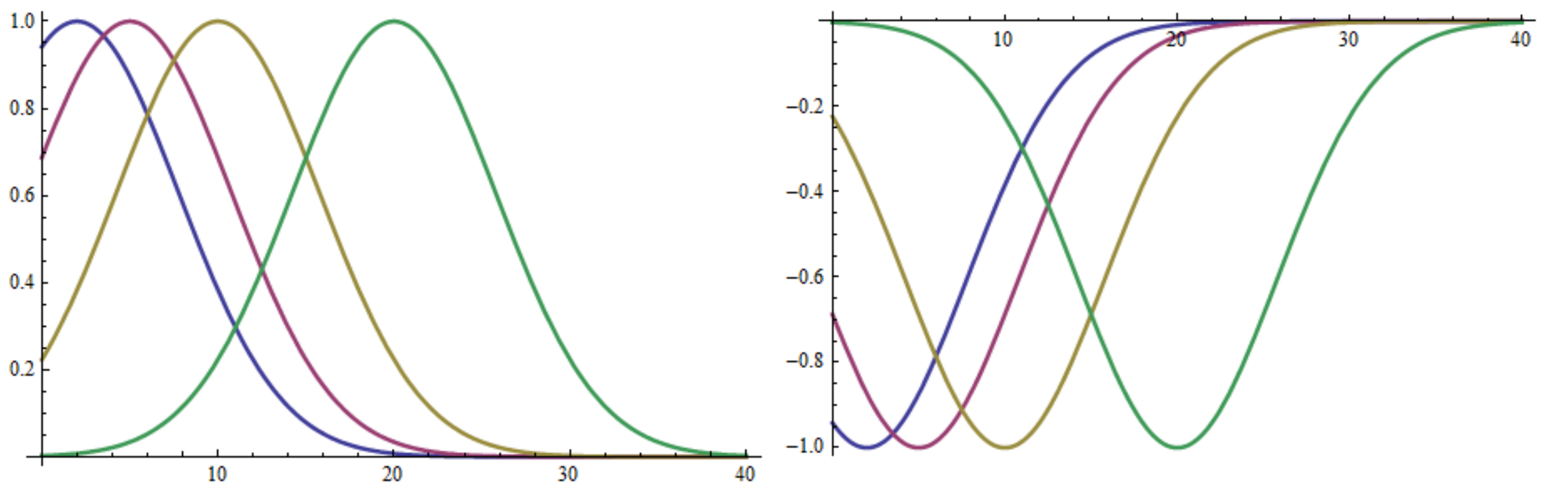}
\caption{Plot of $h(u,x_i)$ for $x_1 = 2$, $x_2 = 5$, $x_3 = 10$ and $x_4=20$, where $b=0.015$ and $c=1$ (left) and $c=-1$ (right).}
\label{hwaves}
\end{center}
\end{figure}
\vspace{-0.2cm}
\noindent The associated bond price and interest rate processes have the following sample paths:

\begin{figure}[H]
\begin{center}
\includegraphics[scale=0.35]{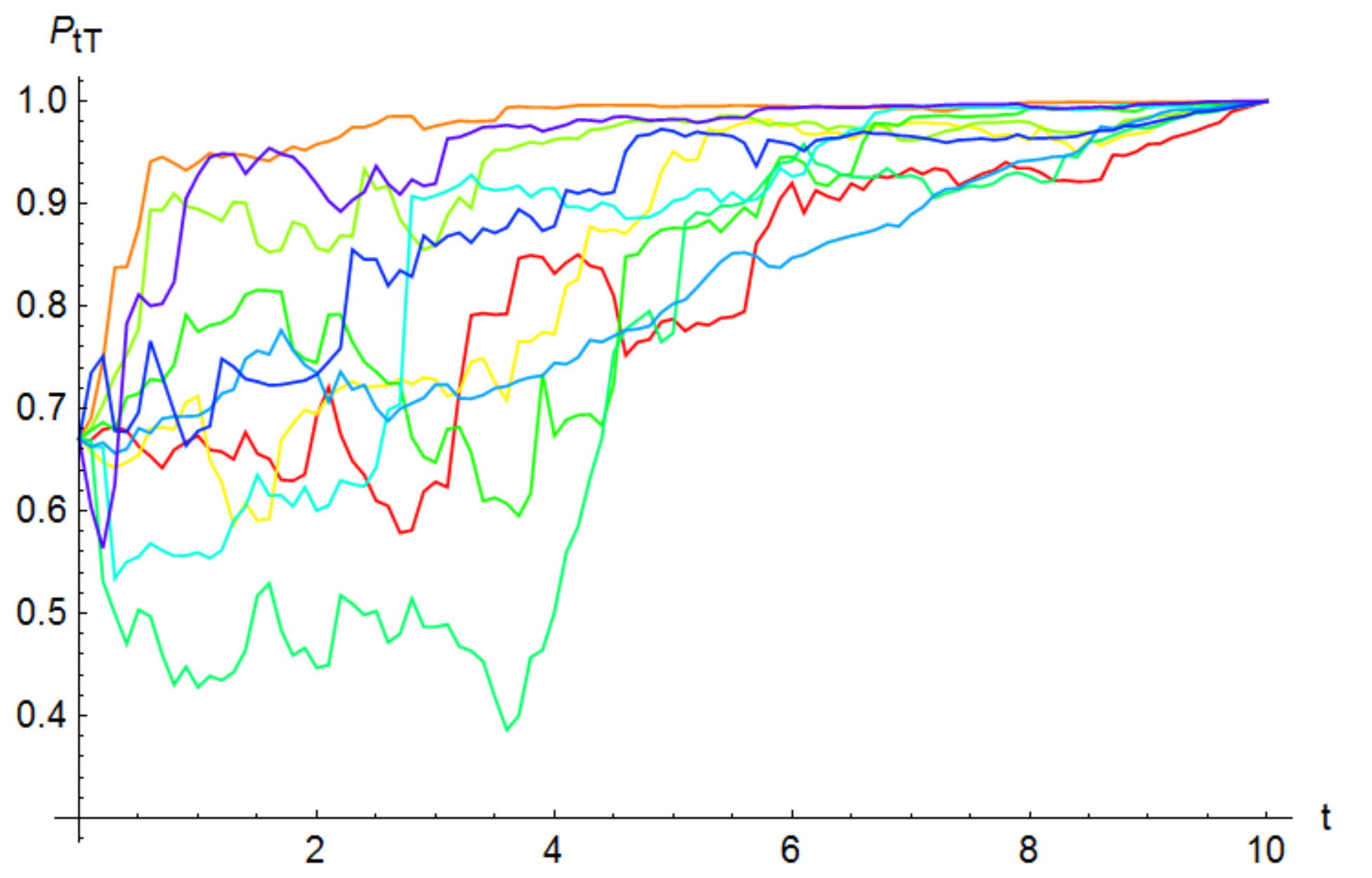}
\includegraphics[scale=0.35]{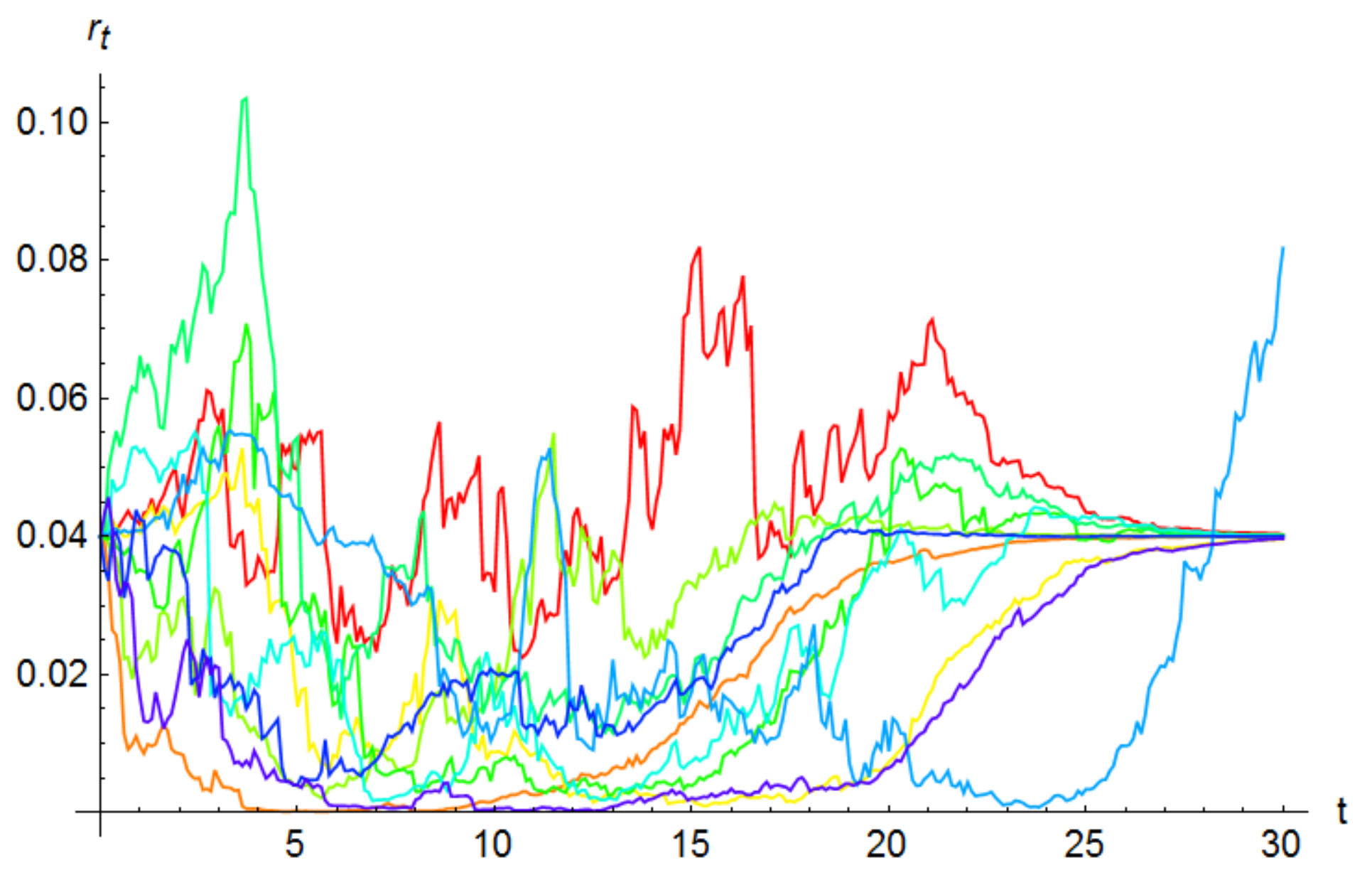}
\caption{Sample paths for a discount bond with $T=10$ and the short rate. We use the variance-gamma model with 
$h(u,X) = c\exp{[-b(u-X)^2]}$. We let $\theta=-1.5$, $\Sigma=2$ and $\nu= 0.25$. We set $f_0(x_1)=0.2$, $f_0(x_2)=0.35$, $f_0(x_3)=0.35$, $f_0(x_4)=0.1$ and $x_1=2$, $x_2=5$, $x_3=10$, $x_4=20$. We choose $\sigma=0.1$, $c=0.5$, $b=0.015$ and the initial term structure is $P_{0t} = \exp{(-0.04t)}$.}
\label{vgbondir}
\end{center}
\end{figure}
\noindent We observe that over time the sample paths of the interest rate process revert to the initial level $r_0$. However, some paths may revert to $r_0$ at a later time than others, depending on the realized value of the random variable $X$.
}
\end{example}


\section{Chameleon random mixers}
The functional form of the random mixer $h(u,X)$ strongly influences the interest rate dynamics. The choice of $h(u,X)$ also affects the robustness of the model: there are choices in which the numerical integration in the calculation of the pricing kernel does not converge. So far, we have constructed examples based on an exponential-type random mixer. However, one may wish to introduce other functional forms for $h(u,X)$ for which we can observe different behaviour in the interest rate dynamics, while maintaining robustness. For instance we may consider a random piecewise function of the form
\begin{equation}\label{piecewise}
h(u,X) = g_1(u)\indi{1}_{\{u\leq X\}} + g_2(u)\indi{1}_{\{u>X\}}
\end{equation}
where  $g_j:\re_+ \rightarrow \re$ for $j=1,2$. The random mixer now has a ``chameleon form'': initially appearing to be $g_1$, and switching its form to $g_2$ at $X=u$. This results in the martingale $\{\hM_{tu}\}$, and the resulting interest rate sample paths, exhibiting different hues over time, depending on the choices of $g_j\;(j=1,2).$ We can extend this idea further by considering (i) multiple $g_j$, or (ii) a multivariate random mixer of the form
\begin{equation}
h(u, X, Y_1, Y_2) =  g_1(u, Y_1)\indi{1}_{\{u\leq X\}} + g_2(u, Y_2)\indi{1}_{\{u>X\}},
\end{equation}
where $X>0$, $Y_1$ and $Y_2$ are independent random variables with associated information processes. In this case, the $g_j$ are themselves random-valued functions. Here $X$ can be regarded as the primary mixer which determines the timing of the regime switch. The variables $Y_i\;(i=1,2)$ can then be interpreted as the secondary mixers determining the weights of the L\'evy processes over two distinct time intervals.

\begin{example}\label{GammaModelwithBrownianInfoChameleon}
\rm{We now present what may be called the ``Brownian-gamma chameleon model''. We consider the filtered gamma martingale family (\ref{gammaMhat}) in the situation where the random mixer $h(u,X)$ has the form
\begin{equation}\label{chamh}
h(u,X) =  c_1 \sin{(\alpha_1u)}\indi{1}_{\{u\leq X\}} +  c_2\exp{\left(-\alpha_2 u\right)}\indi{1}_{\{u>X\}}
\end{equation}
where $c_1, c_2 <\kappa^{-1}$ and $\alpha_2>0$. The information process $\{I_t\}$ associated with $X$ is taken to be of the form
\begin{equation}
I_t = \sigma t X + B_t.
\end{equation}
We assume that $X$ is a positive discrete random variable taking values $\{x_1, x_2, \ldots, x_n\}$ with \textit{a priori} probabilities $f_0(x_i)$, $i=1, 2, \ldots, n$. That is, the function $h(u,X)$ will switch once from sine to exponential behaviour at one of the finitely many random times. Inserting (\ref{gammaMhat}), with the specification (\ref{chamh}), in the expression for the bond price (\ref{fhbp}), we obtain
\begin{equation}\label{chambond}
P_{tT} = \frac{\int_T^\infty \rho(u)\, \sum_{i=1}^n f_t(x_i) \left[1-\kappa \,h(u,x_i)\right]^{mt}\exp{\left[h(u,x_i)\,\gamma_t\right]}\,\rd u}{\int_t^\infty \rho(v)\, \sum_{i=1}^n f_t(y_i) \left[1-\kappa\, h(v,y_i)\right]^{mt}\exp{\left[h(v,y_i)\,\gamma_t\right]}\,\rd v},
\end{equation}
where $h(u,x_i)$ is given by (\ref{chamh}) for $X=x_i$, and 
\begin{equation}
f_t(x_i) = \frac{f_0(x_i)\,\exp\left[\sigma x_i I_t-\tfrac{1}{2}\sigma^2 x_i^2 t\right]}{\sum_{i=1}^n f_0(y_i)\,\exp\left[\sigma y_i I_t -\frac{1}{2}\sigma^2 y_i^2 t\right]}.
\end{equation}
Since the sine function oscillates periodically within the interval $[-1,1]$, the integrals in (\ref{chambond}) may not necessarily converge to one value. However, at some finite random time $u=X$, the sine behaviour is replaced by an exponential decay; this ensures the integrals in the expression for the bond price converge. Such a behaviour may be viewed as a regime switch at a random time. In the simulation below, the analysis of the model parameters is analogous to the one in Example \ref{GammaExamplewithBrownianInfo}. It is worth emphasizing nevertheless that (i) the \textit{a priori} probabilities $f_0(x_i),$\; $i=1,2,\ldots, n$ have a direct influence on the length of the time span during which the sine function in the chameleon mixer is activated, (ii) the magnitude of $\alpha_1$ determines the frequency of the sine wave, while $\alpha_2$ affects the rate at which reversion to the initial interest rate (in the simulation below $r_0 = 4\%$) occurs, and(iii) the size of $c_1$ determines the amplitude of the sine, and it significantly impacts the convergence of the numerical integration. We find that reasonable results are obtained for $-\kappa^{-1} < c_1 < \kappa^{-1}$.

\begin{figure}[H]
\begin{center}
\includegraphics[scale=0.835]{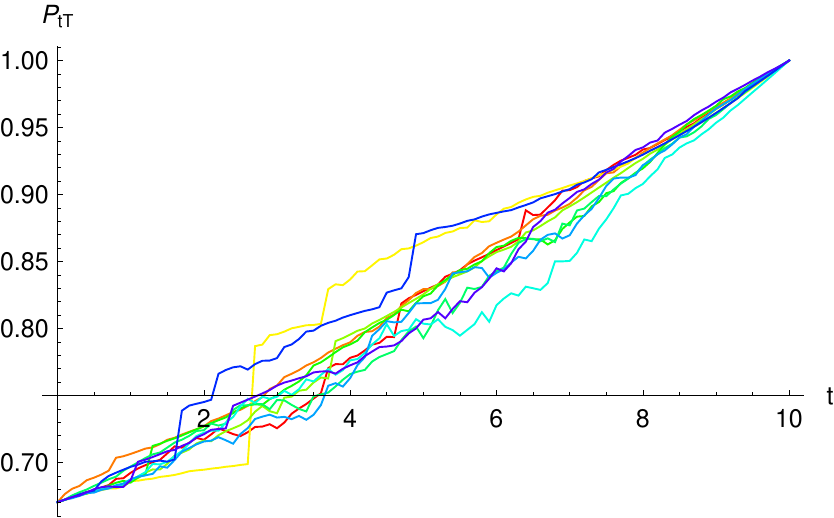}
\includegraphics[scale=0.835]{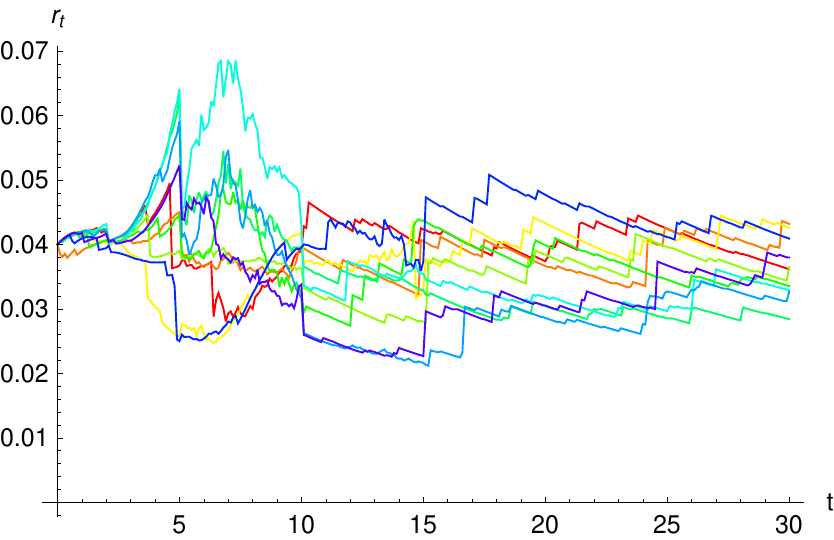}
\caption{Sample paths of discount bond with $T=10$ and short rate trajectories. We use the Brownian-gamma chameleon model with $h(u,X) =   c_1 \sin{(\alpha_1u)}\indi{1}_{\{u\leq X\}} +  c_2\exp{\left(-\alpha_2 u\right)}\indi{1}_{\{u>X\}}$. Let $X$ take the values $\{x_1 = 2,\, x_2 = 5,\, x_3 = 10,\, x_4 = 15\}$ with \textit{a priori} probabilities $\{f_0(x_1) = 0.2,\, f_0(x_2) = 0.35,\, f_0(x_3) = 0.35,\, f_0(x_4) = 0.1\}$. We set $m=0.5$, $\kappa=0.5$, $\sigma=0.1$, $c_1=0.2625$, $c_2=0.75$, $\alpha_1 = 0.75$, $\alpha_2 = 0.02$ and $P_{0t}=\exp{\left(-0.04 t\right)}$.}\vspace{-0.5cm}
\label{chambondir}
\end{center}
\end{figure}
}
\end{example}


\section{Model-generated yield curves}
The yield curve at any time is defined as the range of yields that investors in sovereign debt can expect to receive on investments over various terms to maturity. For a calendar date $t$ and a time to maturity $\tau$, we let $Y_{t, t+\tau}$ be the continuously compounded zero-coupon spot rate for time to maturity $\tau$, that is, the map
$\tau \mapsto Y_{t, t+\tau}$. We write
\begin{equation}
P_{t,t+\tau} = \exp{\left(-\tau Y_{t,t+\tau}\right)}.
\end{equation}
Typically, the following yield curve movements are observed: (i) parallel shifts of the yield curve corresponding to an equal increase in yields across all maturities; (ii) steepening (flattening) of the yield curve, that is the difference between the yields for longer-dated bonds and shorter-dated bonds widens (narrows), and (iii) changes in the curvature and overall shape of the yield curve. The terms ``shift'', ``twist'' and ``butterfly'' are also used to describe these yield curve movements.

As shown in Figure \ref{BGYieldCurves} below, the two-factor Brownian-gamma model set-up in Example \ref{GammaExamplewithBrownianInfo} is indeed too rigid to allow for significant changes in the shape of the yield curve. For $f_0(1)=1$, the yield curve is flat at all times. For $0\leq f_0(1) <1$, this model can generate flat, upward sloping yield curves and in certain cases, slightly inverted yield curves. The variance-gamma model (Figure \ref{VGYieldcurves}) and the Brownian-gamma chameleon model (Figure \ref{ChamBGYieldCurves}) show more flexibility, where changes of slope and different yield curve shapes are observed. These model may generate flat, upward sloping, inverted and humped yield curves. We emphasise that these classes of models are able to capture all three types of yield curve movements.

\begin{figure}[H]
\begin{center}
\includegraphics[scale=0.65]{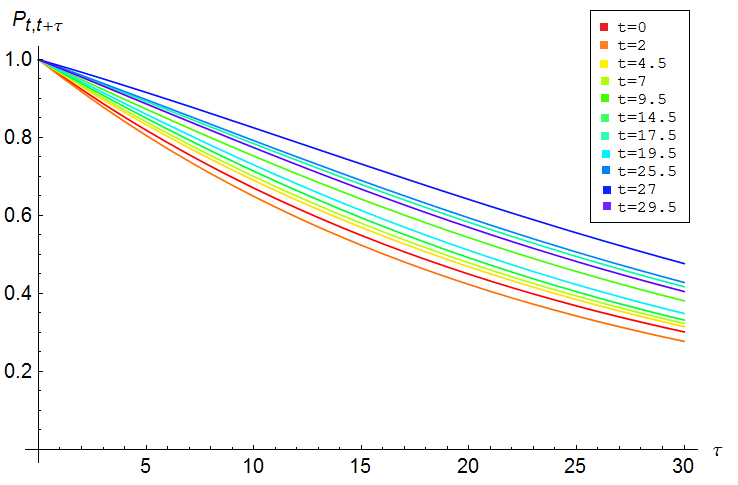}
\vspace{-0.1cm}
\caption{Discount bond curves for the Brownian-gamma model. We let $X= \{0,1\}$ with $f_0(1)=0.3$. We let $m=2$, $\kappa=0.2$, $\sigma = 0.1$, $c=-2$, $b=0.03$, $P_{0t}=\exp{\left(-0.04 t\right)}$.}\label{BGBondCurves}\vspace{-0.5cm}
\end{center}
\end{figure}

\begin{figure}[H]
\begin{center}
\includegraphics[scale=0.65]{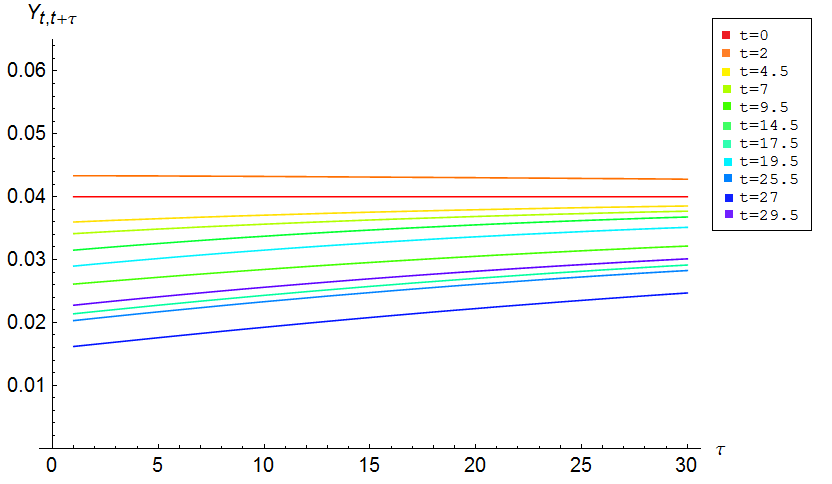}
\vspace{-0.1cm}
\caption{Yield curves for the Brownian-gamma model. We let $X= \{0,1\}$ with $f_0(1)=0.3$. We let $m=2$, $\kappa=0.2$, $\sigma = 0.1$, $c=-2$, $b=0.03$, $P_{0t}=\exp{\left(-0.04 t\right)}$.}\label{BGYieldCurves}\vspace{-0.5cm}
\end{center}
\end{figure}

\begin{figure}[H]
\begin{center}
\includegraphics[scale=0.58]{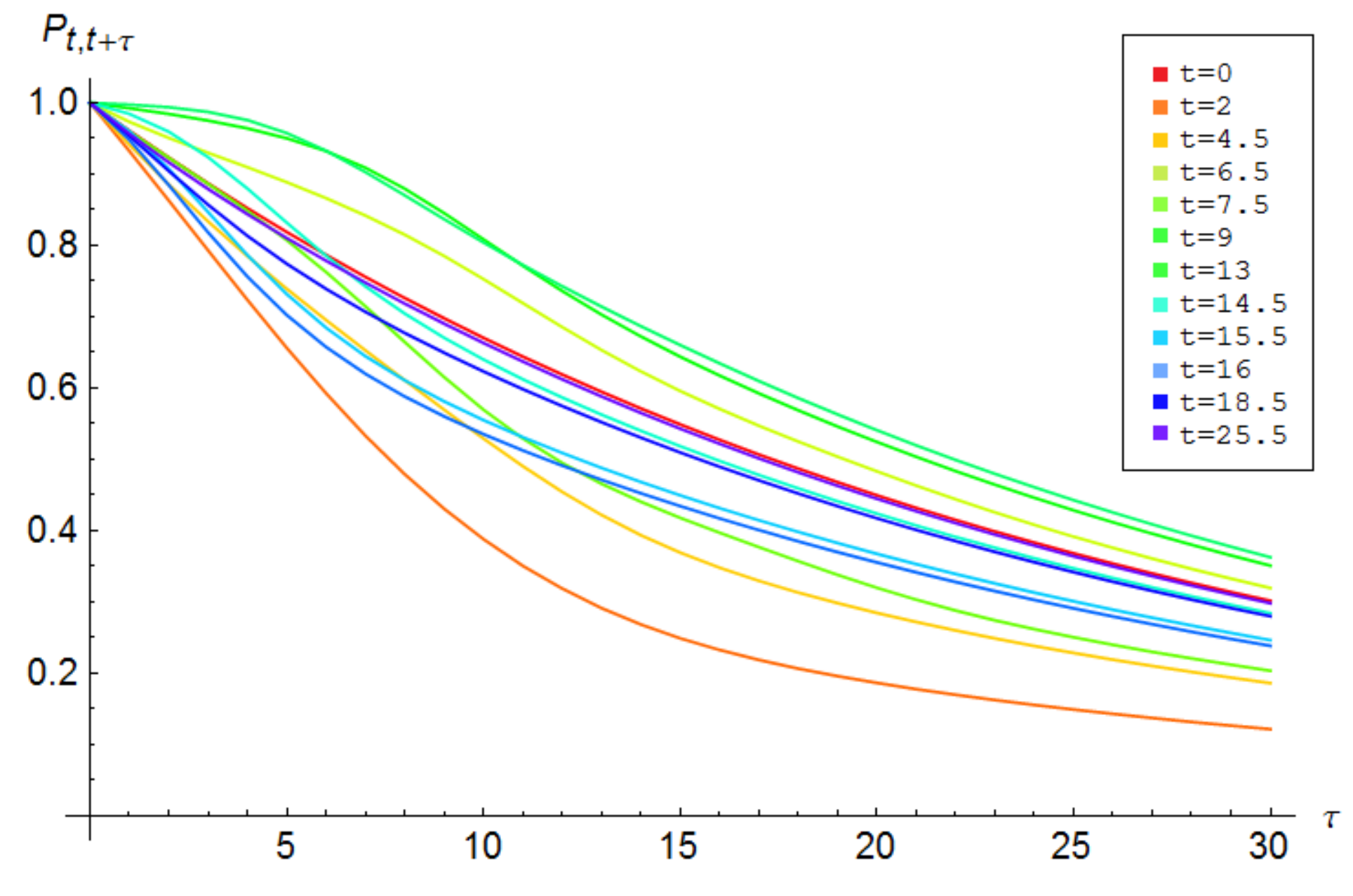}
\caption{Discount bond curves for the variance-gamma model with 
$h(u,X) = c\exp{[-b(u-X)^2]}$. We let $\theta=-1.5$, $\Sigma=2$ and $\nu= 0.25$. We set $f_0(x_1)=0.2$, $f_0(x_2)=0.35$, $f_0(x_3)=0.35$, $f_0(x_4)=0.1$ and $x_1=2$, $x_2=5$, $x_3=10$, $x_4=20$. We choose $\sigma=0.1$, $c=0.5$, $b=0.015$ and the initial term structure is $P_{0t} = \exp{(-0.04t)}$.}
\label{VGBondcurves}
\end{center}
\end{figure}

\begin{figure}[H]
\begin{center}
\includegraphics[scale=0.58]{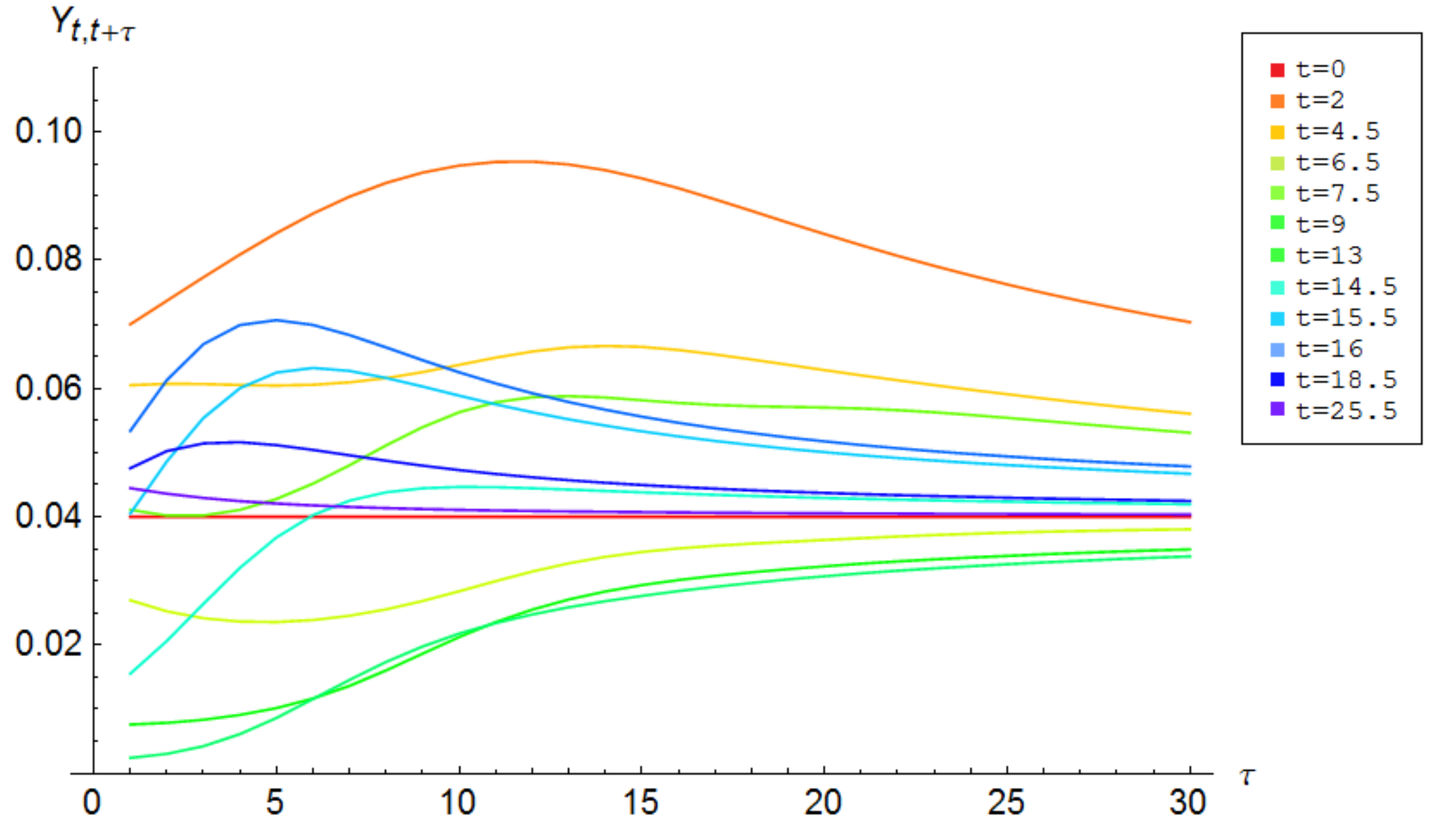}
\caption{Yield curves for the variance-gamma model where 
$h(u,X) = c\exp{[-b(u-X)^2]}$. We let $\theta=-1.5$, $\Sigma=2$ and $\nu= 0.25$. We set $f_0(x_1)=0.2$, $f_0(x_2)=0.35$, $f_0(x_3)=0.35$, $f_0(x_4)=0.1$ and $x_1=2$, $x_2=5$, $x_3=10$, $x_4=20$. We choose $\sigma=0.1$, $c=0.5$, $b=0.015$ and the initial term structure is $P_{0t} = \exp{(-0.04t)}$.}
\label{VGYieldcurves}
\end{center}
\end{figure}

\begin{figure}[H]
\begin{center}
\includegraphics[scale=0.62]{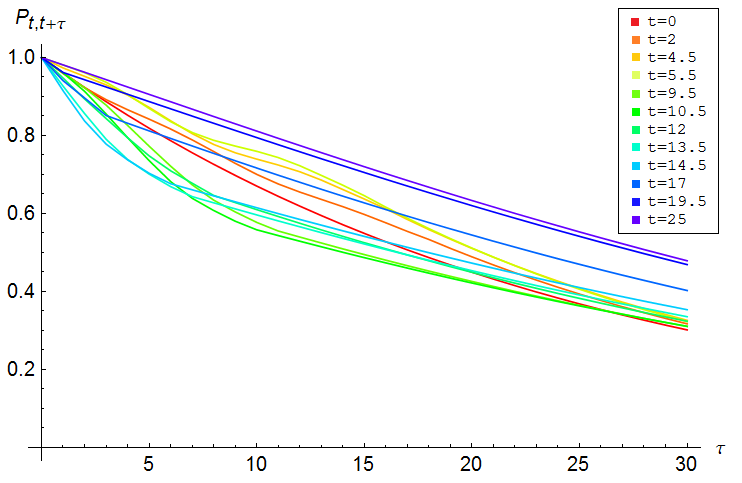}
\caption{Discount bond curves for the Brownian-gamma chameleon model. We let $X= \{x_1=2, x_2=5, x_3=10, x_4=20\}$ with $f_0(x_1)=0.15$, $f_0(x_2)=0.35$, $f_0(x_3)=0.35$, $f_0(x_4)=0.15$. We let $m=0.5$, $\kappa=0.5$, $\sigma = 0.1$, $c_1=-0.4375$, $c_2=-1.25$, $\alpha_1=0.75$, $\alpha_2=0.02$, $P_{0t}=\exp{\left(-0.04 t\right)}$.}
\label{ChamBGBondCurves}\vspace{-0.5cm}
\end{center}
\end{figure}

\begin{figure}[H]
\begin{center}
\includegraphics[scale=0.62]{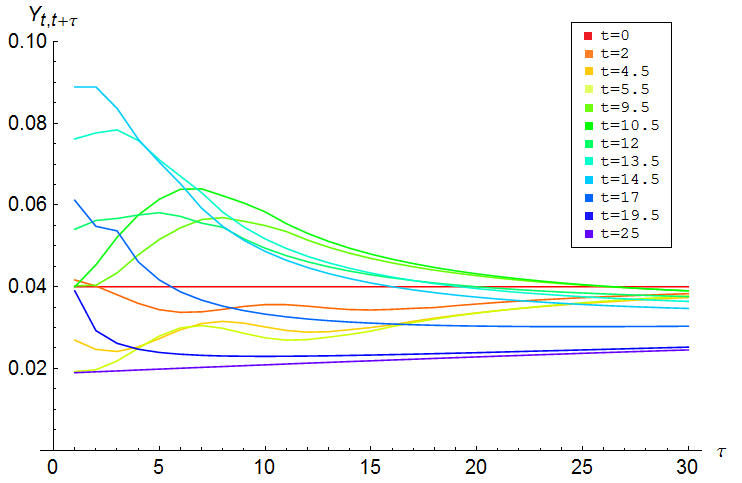}
\caption{Yield curves for the Brownian-gamma chameleon model. We let $X= \{x_1=2, x_2=5, x_3=10, x_4=20\}$ with $f_0(x_1)=0.15$, $f_0(x_2)=0.35$, $f_0(x_3)=0.35$, $f_0(x_4)=0.15$. We let $m=0.5$, $\kappa=0.5$, $\sigma = 0.1$, $c_1=-0.4375$, $c_2=-1.25$, $\alpha_1=0.75$, $\alpha_2=0.02$, $P_{0t}=\exp{\left(-0.04 t\right)}$.}
\label{ChamBGYieldCurves}\vspace{-0.5cm}
\end{center}
\end{figure}


\section{Pricing of European-style bond options}
Let $\{C_{st}\}_{0\leq s\leq t<T}$ be the price process of a European call option with maturity $t$ and strike $0<K<1$, written on a discount bond with price process $\{P_{tT}\}_{0\leq t\le T}$. The price of the option at time $s$ is given by
\begin{equation}\label{call}
C_{st} = \frac{1}{\pi_s}\,\EP\left[\pi_t (P_{tT}-K)^+\,|\,\mathcal{F}_s\right].
\end{equation}
By substituting (\ref{fhpk}) and (\ref{fhbp}) into (\ref{call}), we obtain
\begin{equation}
C_{st} = \frac{1}{\pi_s}\EP\left[\left(\int_T^\infty \rho(u)\, \hM_{tu} \,\rd u - K \int_t^ \infty \rho(u)\, \hM_{tu}\, \rd u\right)^+\,\bigg|\, \cF_s\right].
\end{equation}
In the single-factor models that we have considered with a Markovian information process $\{I_t\}$ , we can define the region $\mathcal{V}$ by
\begin{equation}
\mathcal{V} := \left\{y, z: \int_T^\infty \rho(u)\, \hM_{tu}(L_t = y, I_t = z)\,\rd u - K\int_t^\infty \rho(u)\, \hM_{tu}(L_t = y, I_t = z)\,\rd u > 0 \right\}.
\end{equation}
It follows that the price of the call option is 
\begin{equation}
C_{st} = \frac{1}{\pi_s}\int\int_{\mathcal{V}} \left(\int_T^\infty \rho(u)\, \hM_{tu}(y, z)\,\rd u - K\int_t^\infty \rho(u)\, \hM_{tu}(y, z)\,\rd u \right) q_s(y, z)\, \rd y\, \rd z
\end{equation}
where
\begin{eqnarray}
q_s(y, z) = \frac{\partial^2}{\partial y\,\partial z} \pr \left[L_t \leq y, I_t \leq z\, |\, \cF_s\right].
\end{eqnarray}
We can use Fubini's theorem to write this more compactly in the form
\begin{equation}
C_{st} = \frac{1}{\pi_s}\left(\int_T^\infty \rho(u)\,\Phi_{tu}\,\rd u - K\int_t^\infty \rho(u)\,\Phi_{tu}\,\rd u\right),
\end{equation}
where 
\begin{equation}
\Phi_{tu} = \int \int_\mathcal{V} \hM_{tu}(y,z)\,q_s(y, z)\,\rd y\, \rd z.
\end{equation}
We apply Monte Carlo techniques to simulate option price surfaces. A large number of iterations is required to obtain accurate estimates. To increase precision, variance reduction techniques or quasi-Monte Carlo methods can be considered (see Boyle {\it et al}.~\cite{bbg}). The choice of the random mixer affects the shape of the resulting option price surface. The simulations in Figure \ref{Surfaces} are based on (i) the Brownian-gamma model constructed in Example \ref{GammaExamplewithBrownianInfo}, and (ii) the Brownian-gamma chameleon model in Example \ref{GammaModelwithBrownianInfoChameleon}. The wave across the second option price surface is produced by the sine function that defines part of the chameleon random mixer.
\begin{figure}[H]
\includegraphics[scale=0.56]{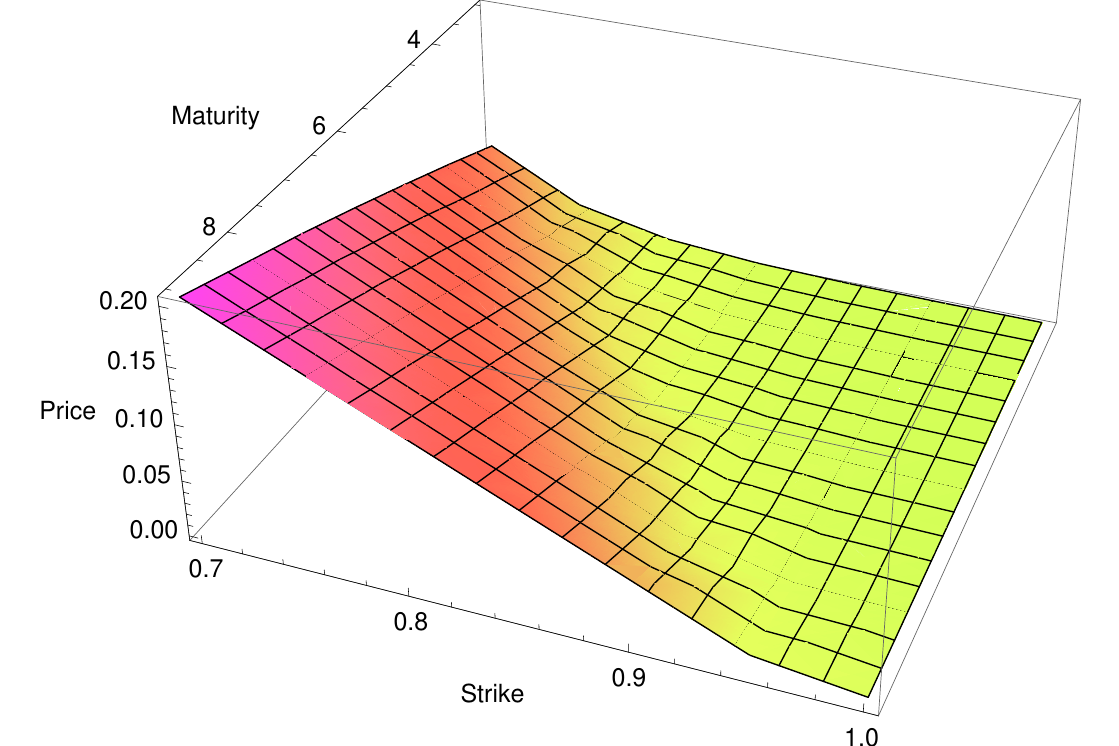}
\includegraphics[scale=0.56]{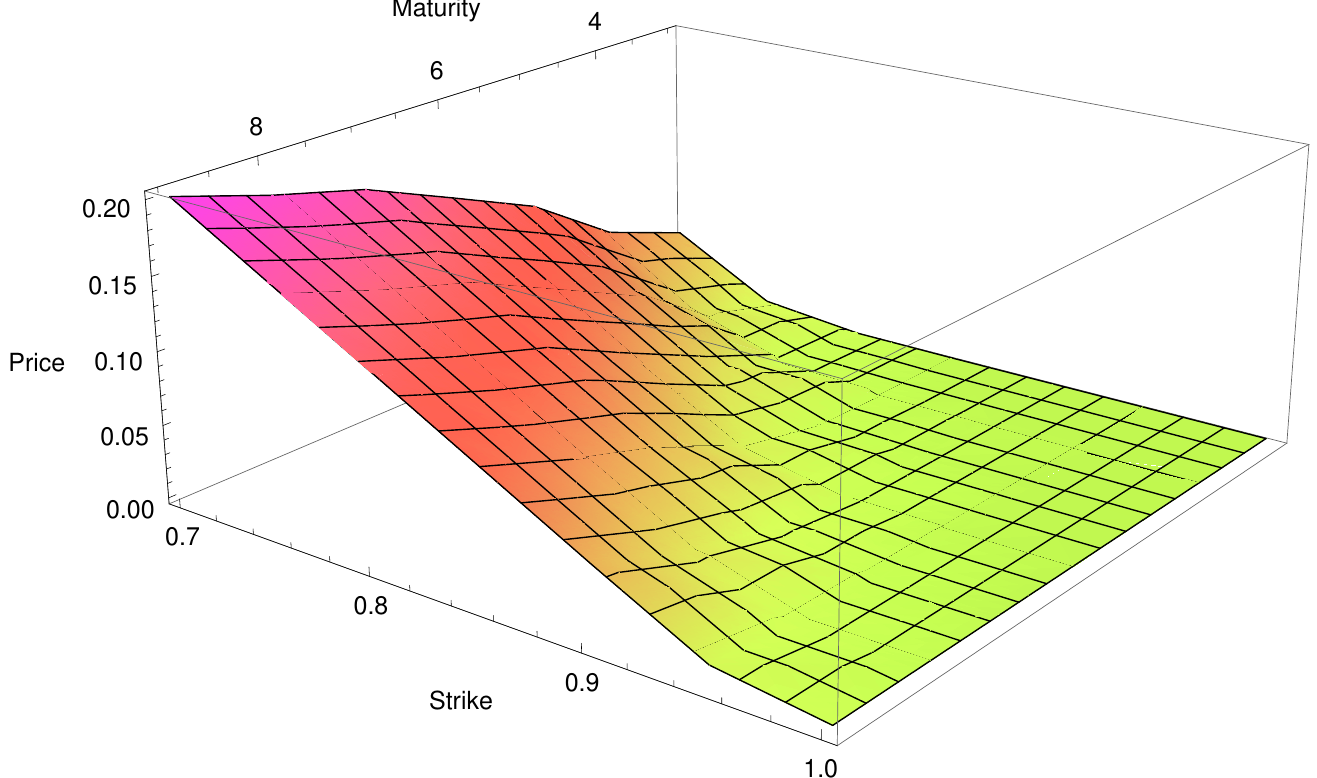}
\begin{center}
\caption{Option price surface at $s=2$ of call options on a discount bond with $T=10$. \;\;\; \;\;\; $(i)$ Simulation based on the Brownian-gamma model. We set $X=\{0,1\}$ with $f_0(1)=0.5$, $m=0.5$, $\kappa=0.5$, $\sigma=0.1$, $c=-2$, $b=0.03$ and $P_{0t}=\exp{\left(-0.04 t\right)}$. $(ii)$ Simulation based on the Brownian-gamma chameleon model. We set $X=\{x_1=2, x_2 = 5, x_3 = 10, x_4 = 20\}$ with $f_0(x_1)=0.15$, $f_0(x_2)=0.35$, $f_0(x_3)=0.35$, $f_0(x_4)=0.15$, $m=0.5$, $\kappa=0.5$, $\sigma=0.1$, $c_1= 0.35$, $c_2=1$, $\alpha_1 =3$, $\alpha_2=0.03$, and $P_{0t}=\exp{\left(-0.04 t\right)}$. }\label{Surfaces}\vspace{-0.5cm}
\end{center}
\end{figure}


\section{Randomised heat kernel interest rate models}
In Sections 2 and 3, we constructed martingales based on L\'evy processes and an Esscher-type formulation. We
recall that the pricing kernel is modelled by
\begin{eqnarray}
\pi_t &=& \int_t^\infty \rho(u)\,\EP\left[M_{tu}\left(X, L_t\right)\,|\,\cF_t\right]\,\rd u\nn\\
&=& \int_{-\infty}^\infty {\int_t^\infty \rho(u)\,M_{tu}\left(x, L_t\right)\,\rd u}\,f_t(x)\,\rd x.
\end{eqnarray}
The process $\{M_{tu}\left(X, L_t\right)\}$ is a unit-initialized positive $\{\mathcal{G}_t\}$-martingale, and the process 
\begin{equation}
S_t\left(X, L_t\right) := \int_t^\infty \rho(u)\, M_{tu}\left(X, L_t\right)\,\rd u
\end{equation}
is a positive $\{\mathcal{G}_t\}$-supermartingale. The projection of a positive $\{\mathcal{G}_t\}$-supermartingale onto $\{\mathcal{F}_t\}$, that is
\begin{equation}
\pi_t := \EP\left[S_t\left(X, L_t\right)\,|\,\cF_t\right],
\end{equation} 
is an $\{\mathcal{F}_t\}$-supermartingale (F\"ollmer \& Protter \cite{fp}, Theorem 3). 
\\\\
\noindent \textbf{Weighted heat kernel approach.}\; We now model the impact of uncertainty on a financial market by a process that has the Markov property with respect to its natural filtration, and which we denote $\{Y_t\}_{t\geq 0}$. Of course, the case where $\{Y_t\}$ is a L\'evy process, which is a Markov process of Feller type, is included (see Applebaum \cite{app}). 
 
\newtheorem{defprop2}[definition]{Definition}
\begin{defprop2}
Let $\{Y_t\}$ be a Markov process with respect to its natural filtration. A measurable function $p:\re_+\times\re_+\times\re \rightarrow \re$ is a propagator if it satisfies
\begin{equation}
\mathbb{E}\left[p\left(t, v, Y_t\right)\,|\,Y_s \right] = p\left(s, v+t-s, Y_s\right)
\end{equation}
for $(v,t) \in \re_+\times\re_+$ and $0\leq s\leq t$.\\
\end{defprop2}
Next, let $\{n_t\}_{t\geq 0}$ be a pure noise process, and let the filtration $\{\mathcal{G}_t\}$ be generated by
\begin{eqnarray}
\mathcal{G}_t = \sigma\left(\{Y_s\}_{0\leq s\leq t}, \{n_s\}_{0\leq s\leq t},\, X\right),
\end{eqnarray}
where $\{Y_t\}$, $\{n_t\}$, and the random variable $X$ are all independent. Let $G(\cdot)$ be a positive bounded function\footnote{Once a Markov process $\{Y_t\}$ has been chosen, it may be sufficient to relax the boundedness condition, and choose $G(\cdot)$ to be a positive and integrable function.}, and let $h:\re_+\times \re \rightarrow\re.$ Then we set
\begin{equation}\label{propagator}
p(t, v, Y_t, X) := \EP\left[G\left(h(t+v, X), Y_{t+v}\right)\,|\,\mathcal{G}_t\right].
\end{equation}
This is a $\{\mathcal{G}_t\}$-propagator since $X$ is $\mathcal{G}_0$-measurable. It follows that
\begin{equation}
S_t(X, Y_t) := \int_0^\infty w(t,v)\,\EP\left[G(h(t+v, X), Y_{t+v})\,|\, \mathcal{G}_t\right]\,\rd v\label{five}
\end{equation}
is a $\{\mathcal{G}_t\}$-supermartingale, see Akahori {\it et al.}~\cite{ahtt}. Here $w(t,v)$ is a positive function that satisfies
\begin{equation}\label{weightineq}
w(t,v-s) \leq w(t-s,v)
\end{equation}
for arbitrary $t,v\in\re_+$ and $s\leq t\wedge v$. Now we define the market filtration $\{\cF_t\}$ by 
\begin{equation}
\cF_t = \sigma\left(\{Y_s\}_{0\leq s\leq t}, \{I_s\}_{0\leq s\leq t}\right),
\end{equation}
where $\{I_t\}$ carries information about $X$, which is distorted by the pure noise $\{n_t\}$. We have that $\cF_t \subset \mathcal{G}_t$. Then, by F\"ollmer \& Protter \cite{fp} Theorem 3, the projection
\begin{equation}
\pi_t := \EP\left[S_t(X, Y_t)\,|\,\cF_t\right]
\end{equation}
is an $\{\mathcal{F}_t\}$-supermartingale. It follows that
\begin{eqnarray}
\pi_t &=& \EP\left[\int_0^\infty w(t,v)\,\EP\left[G\left(h(t+v, X), Y_{t+v}\right)\,|\,\mathcal{G}_t\right]\,\rd v\,\bigg|\,\cF_t\right],\nn\\
&=& \int_0^\infty w(t,v)\,\EP\left[\EP\left[G(h(t+v, X), Y_{t+v})\,|\mathcal{G}_t\right]\,|\,\cF_t\right]\,\rd v,\nn\\
&=& \int_0^\infty w(t,v)\,\EP\left[G(h(t+v, X), Y_{t+v})\,|\,\cF_t\right]\,\rd v.\label{pkwhkafil}
\end{eqnarray}
We emphasize that in equation (\ref{pkwhkafil}), $\EP\left[G(h(t+v, X), Y_{t+v})\,|\,\cF_t\right]$ is not an $\{\cF_t\}$-propagator when $\{I_t\}$ is not a Markov process. Nevertheless, $\{\pi_t\}$ is a valid model for the pricing kernel, subject to regularity conditions.


\section{Quadratic model based on the Ornstein-Uhlenbeck process}
In this section, we generate term structure models by using Markov processes with dependent increments. We emphasize that such models cannot be constructed based on the filtered Esscher martingales. Let us suppose that $\{Y_t\}$ is an Ornstein-Uhlenbeck process with dynamics 
\begin{equation}
\rd Y_t = \delta (\beta - Y_t)\,\rd t + \Upsilon\,\rd W_t,
\end{equation}
where $\delta$ is the speed of reversion, $\beta$ is the long-run equilibrium value of the process and $\Upsilon$ is the volatility.
Then, for $s\leq t$, the conditional mean and conditional variance are given by
\begin{eqnarray}
\EP\left[Y_t\,|\,Y_s\right] &=& Y_s\,\exp{\left[-\delta(t-s)\right]} +\beta\left(1-\exp{\left[-\delta(t-s)\right]}\right).\label{condmean}\\
\textup{Var}\left[Y_t\,|\,Y_s\right] &=& \frac{\Upsilon^2}{2\delta}\left(1-\exp{\left[-2\delta(t-s)\right]}\right). \label{condvar}
\end{eqnarray}
Let us suppose, for a well-defined positive function $h:\re_+\times\re \rightarrow \re_+$, that
\begin{equation}
G(h(v, X), Y_{v}) = h(v,X) Y^2_v.
\end{equation}
Since $X$ is $\mathcal{G}_0$-measurable, and by applying (\ref{condmean}) and (\ref{condvar}), it follows that
\begin{eqnarray}
p(u, t, Y_t, X) &=& \EP\left[h(t+u, X) Y^2_{t+u}\,|\,\mathcal{G}_t\right],\nn\\
&=& h(t+u, X) \,\EP\left[\left(Y_{t+u}-\EP\left[Y_{t+u}\,|\,Y_t\right] + \EP\left[Y_{t+u}\,|\,Y_t\right]\right)^2\,|\,Y_t\right],\nn\\
&=& h(t+u, X)\,\left[\textup{Var}\left[Y_{t+u}\,|\,Y_t\right] +\EP\left[Y_{t+u}\,|\,Y_t\right]^2\right],\nn\\
&=& h(t+u, X)\,\left[\frac{\Upsilon^2}{2\delta}\left(1-\textup{e}^{-2\delta u}\right) + \left[Y_t\,\textup{e}^{-\delta u} +\beta\left(1-\textup{e}^{-\delta u}\right)\right]^2\right].\nn\\
\end{eqnarray}
The pricing kernel is then given by (\ref{pkwhkafil}), and we obtain
\begin{align}\label{whkapiexpanded}
\pi_t = \int_0^\infty w(t,u) \left[\frac{\Upsilon^2}{2\delta}\left(1-\textup{e}^{-2\delta u}\right) + \left[Y_t\,\textup{e}^{-\delta u} +\beta\left(1-\textup{e}^{-\delta u}\right)\right]^2\right]\nn\\
\times \int_{-\infty}^\infty h(t+u,x) f_t(x) \,\rd x\,\rd u.
\end{align}
It follows that the price of a discount bond is expressed by
\begin{eqnarray}
P_{tT} = \frac{1}{\pi_t}\,\EP\left[\int_0^\infty w(T,v) \,\EP\left[G(h\left(T+v,X),Y_{T+v}\right)\,|\,\cF_T\right]\,\rd v\;\bigg|\,\cF_t\right],
\end{eqnarray}
where $\{\pi_t\}$ is given in (\ref{whkapiexpanded}), and the conditional expectation can be computed to obtain
\begin{align}
\int_0^\infty w(T,v) \, \left[\frac{\Upsilon^2}{2\delta}\left(1-\textup{e}^{-2\delta (T+v-t)}\right) + \left[Y_t\,\textup{e}^{-\delta (T+v-t)} +\beta\left(1-\textup{e}^{-\delta (T+v-t)}\right)\right]^2\right]\nn\\
\times \int_{-\infty}^\infty h(T+v,x) f_t(x) \,\rd x\,\rd v.
\end{align}

\begin{example}
{\rm
We assume that $X$ is a positive random variable that takes discrete values $\{x_1, \ldots, x_n\}$ with \textit{a priori} probabilities $\{f_0(x_1), \ldots, f_0(x_n)\}$. We suppose that the information flow $\{I_t\}$ is governed by
\begin{equation}
I_t = \sigma X t + B_t.
\end{equation}
We choose the random mixer to be
\begin{equation}
h(t+u, X) = c_1\exp{\left[-c_2(t+u-X)\right]}(t+u),
\end{equation}
where $c_1>0$ and $c_2>0$, and we assume that the weight function is
\begin{equation}
w(t,u) = \exp\left[{-j(u+t)}\right]
\end{equation}
for $j>0$. Later, in Proposition \ref{fhwhka}, we show that this model belongs to the Flesaker-Hughston class. Therefore, the short rate of interest takes the form
\begin{eqnarray}
r_t &=& \frac{e^{-jt}\,\EP\left[G(h(t,X),Y_t)\,|\,\cF_t\right]}{\int_0^\infty e^{-j(t+v)}\, \EP\left[G(h(t+v,X),Y_{t+v})\,|\,\cF_t\right]\,\rd v}.
\end{eqnarray}

Next we simulate the trajectories of the discount bond and the short rate process. We refer to Iacus \cite{iac} for the simulation of the Ornstein-Uhlenbeck process using an Euler scheme. We observe oscillations in the sample paths owing to the mean-reversion in the Markov process.

\begin{figure}[H]
\begin{center}
\includegraphics[scale=0.35]{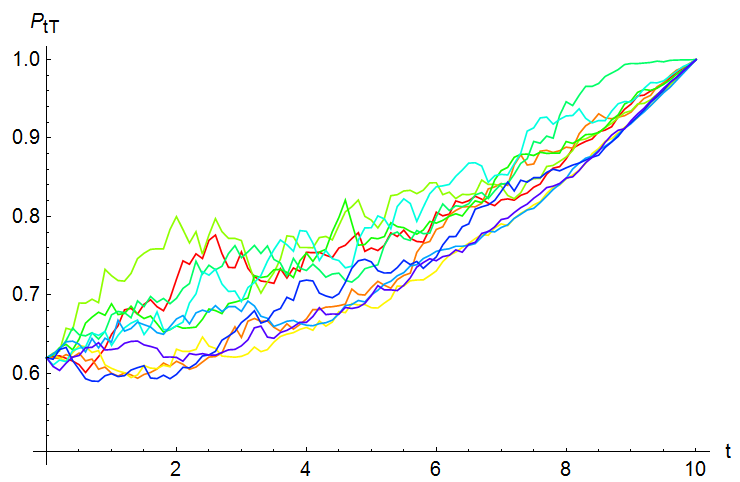}
\includegraphics[scale=0.35]{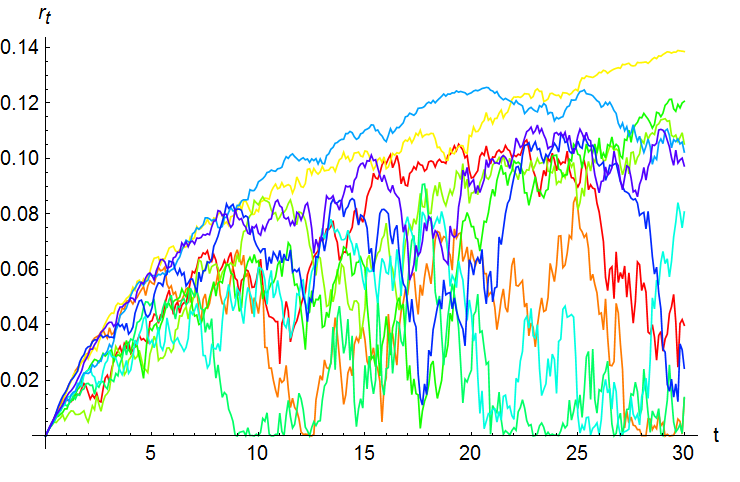}
\caption{Sample paths for a discount bond with $T=10$ and the short rate for the quadratic OU-Brownian model with 
$h(t+u,X) = c_1\exp{\left(-c_2(t+u-X)\right)}(t+u)$ with $c_1=0.02$ and $c_2=0.1$. We let $\delta=0.02$, $\beta=0.5$, $\Upsilon=0.2$ and $Y_0=1$. We let $x_1=1$ and $x_2=2$ where $f_0(x_1)=0.3$ and $f_0(x_2)=0.7$ and $\sigma=0.1$. The weight function is given by $w(t,u) = \exp{\left[-0.04(t+u)\right]}$.
}
\label{OUbondir}
\end{center}
\end{figure}
\noindent
The model-generated yield curves follow. In this example, we mostly observe changes of slope and shifts. However, it should be possible to produce changes of shape in the yield curve by varying the choices of $G(\cdot)$ and $h(\cdot)$.
\begin{figure}[H]
\begin{center}
\includegraphics[scale=0.55]{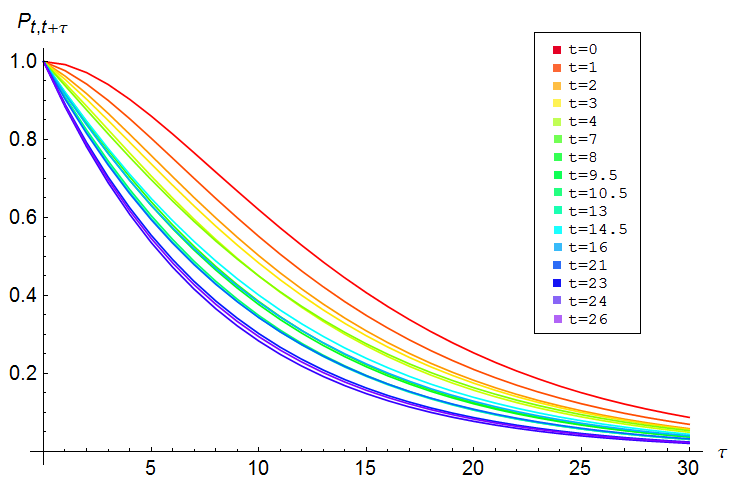}
\caption{Discount bond curves for the quadratic OU-Brownian model with 
$h(t+u,X) = c_1\exp{\left(-c_2(t+u-X)\right)}(t+u)$ with $c_1=0.01$ and $c_2=0.1$. We let $\delta=0.02$, $\beta=0.5$, $\Upsilon=0.2$ and $Y_0=1$. We let $x_1=1$ and $x_2=2$ where $f_0(x_1)=0.5$ and $f_0(x_2)=0.5$ and $\sigma=0.1$. The weight function is given by $w(t,u) = \exp{\left[-0.04(t+u)\right]}$.}
\label{OUBondcurve}
\end{center}
\end{figure}

\begin{figure}[H]
\begin{center}
\includegraphics[scale=0.55]{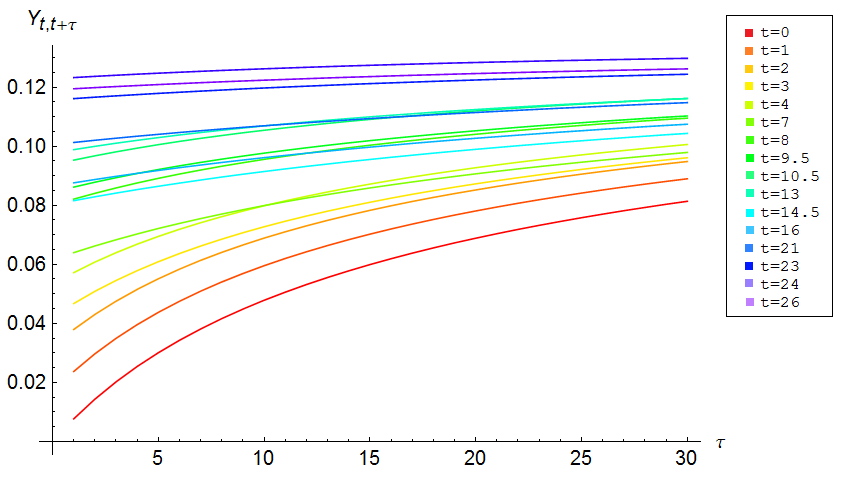}
\caption{Yield curves for the quadratic OU-Brownian model with 
$h(t+u,X) = c_1\exp{\left(-c_2(t+u-X)\right)}(t+u)$ with $c_1=0.01$ and $c_2=0.1$. We let $\delta=0.02$, $\beta=0.5$, $\Upsilon=0.2$ and $Y_0=1$. We let $x_1=1$ and $x_2=2$ where $f_0(x_1)=0.5$ and $f_0(x_2)=0.5$ and $\sigma=0.1$. The weight function is given by $w(t,u) = \exp{\left[-0.04(t+u)\right]}$.}
\label{OUYieldcurve}
\end{center}
\end{figure}
}
\end{example}
\vspace{-0.2cm}


\section{Classification of interest rate models}
In what follows, we show that, under certain conditions, the constructed pricing kernels based on weighted heat kernel models belong to the Flesaker-Hughston class.

\newtheorem{propfhwhka}[proposition]{Proposition}
\begin{propfhwhka}
\label{fhwhka}
Let $\{Y_t\}$ be a Markov process, and let the weight function be given by $w(t,v) = \psi(t+v)$ with $\psi:\re_+\rightarrow \re_+$ such that
\begin{equation}
\int_0^\infty \psi(s)\, \EP\left[G\left(h(s,X), L_{s}\right)\right]\, \rd s \;<\,\infty.
\end{equation} 
Then, the pricing kernel
\begin{equation}
\pi_t = \int_0^\infty \psi(t+v) \,\EP\left[G\left(h(t+v,X), L_{t+v}\right)\,|\,\cF_t\right]\,\rd v
\end{equation}
is a potential generated by 
\begin{equation}
A_t = \int_0^t \psi(u)\, G\left(h(u,X), L_{u}\right)\, \rd u, 
\end{equation}
that is, a potential of class (D). Thus, the pricing kernel is of the Flesaker-Hughston type. 
\end{propfhwhka}
\begin{proof}
The function $w(t,v) = \psi(t+v)$ satisfies (\ref{weightineq}), and thus is a weight function. Then we see that
\begin{eqnarray}
\pi_t &=& \int_t^\infty \psi(u)\,\EP\left[G\left(h(u,X),L_u\right)\right] \, m_{tu}\,\rd u\nn\\
&=& \pi_0\,\int_t^\infty \rho(u)\,m_{tu}\,\rd u,
\end{eqnarray}
where
\begin{equation}
m_{tu} = \frac{\EP\left[G\left(h(u,X),L_u\right)\,|\,\cF_t\right]}{\EP\left[G\left(h(u,X),L_u\right)\right]}
\end{equation}
is a positive unit-initialized $\{\mathcal{F}_t\}$-martingale for each fixed $u\geq t$. The constant $\pi_0$ is a scaling factor. 
\end{proof}

We note that, for instance, the potential models of Rogers \cite{rog} which can be generated by the weighted heat kernel approach with $\psi(t+v) = \exp{\left[-\alpha(t+v)\right]}$ where $\alpha>0$, are Flesaker-Hughston models.
To generate potentials from the weighted heat kernel approach with a general weight $w(t,v)$, the weight function and $G(\cdot)$ should be chosen so that $\EP[\pi_t] \rightarrow 0$ as $t \rightarrow \infty$.
\\
\indent Let us suppose that $\{Y_t\}$ is a Markov process with independent increments. Then the class of Esscher-type randomised mixture models presented in this paper, for which
\begin{equation}
M_{tu}(X,L_t) := \frac{\exp{\left[h(u,X)L_{t}\right]}}{\EP\left[\exp{\left[h(u, X)L_{t}\right]}\,|\,X\right]},
\end{equation}
cannot be constructed by using the weighted heat kernel approach. We see this by setting
\begin{equation}
G\left(h(v, X), L_{t+v}\right) = \frac{\exp{\left[h(v,X)L_{t+v}\right]}}{\EP\left[\exp{\left[h(v, X)L_{t+v}\right]}\,|\,X\right]},
\end{equation}
and by observing that $\EP[G(h(v,X), L_{t+v})\,|\,\mathcal{G}_t]$ is not a $\{\mathcal{G}_t\}$-propagator. As we mentioned earlier, the class of models introduced by Brody {\it et al.}~\cite{bhmack} is included in the class of Esscher-type randomised mixture models. Similarly, models based on kernel functions of the form $G(h(x), Y_t)$ can produce other Esscher-type models by use of the weighted heat kernel approach. The following is a diagrammatic representation of the considered classes of positive interest rate models:
\begin{figure}[H]
\begin{center}
\includegraphics[scale=0.60]{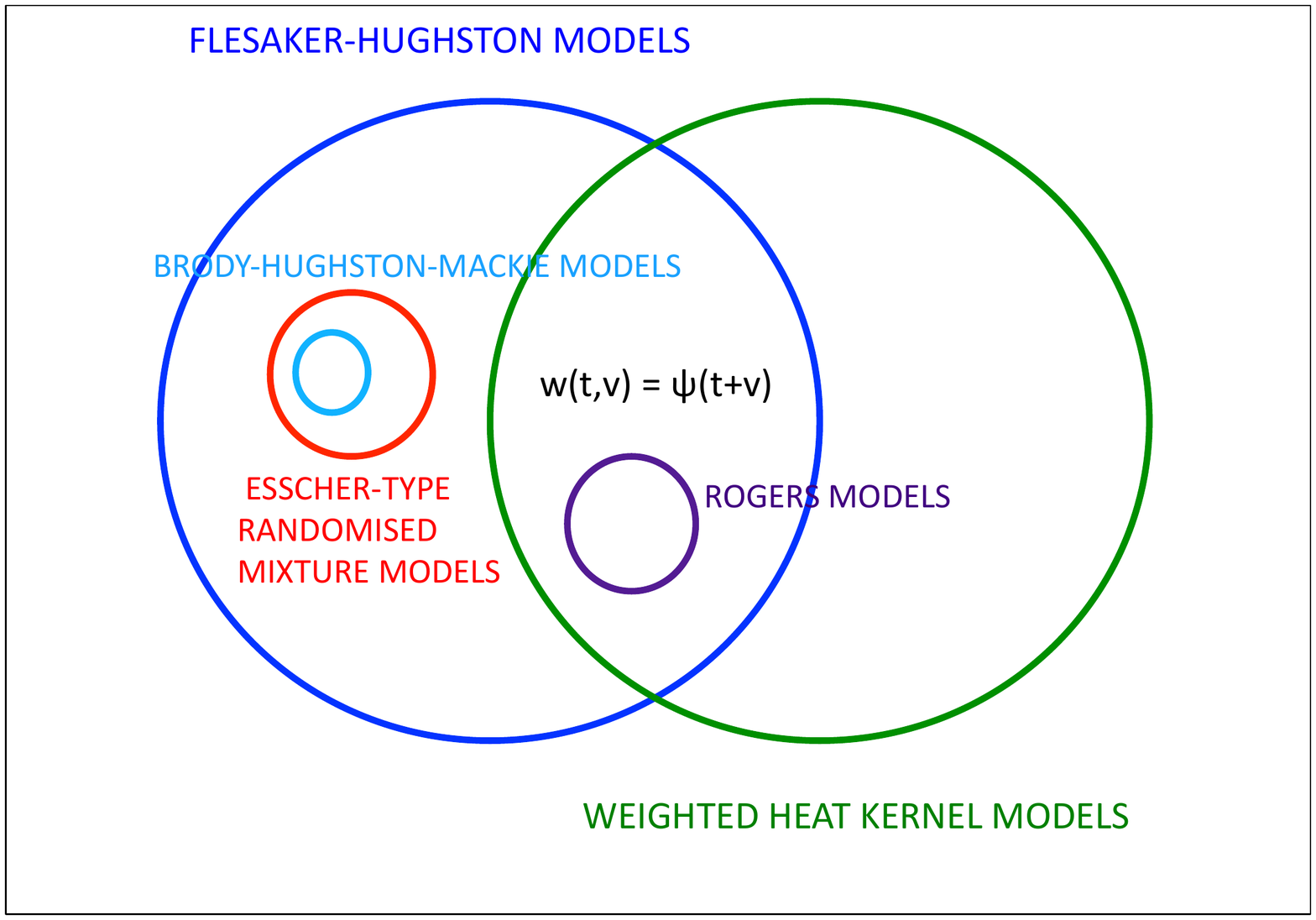}
\vspace{-0.1cm}
\caption{Classes of interest rate models.}
\end{center}
\end{figure}


We conclude with the following observations. The pricing kernel models proposed in this paper are versatile by construction, and potentially allow for many more investigations. For instance, we can think of applications to the modelling of foreign exchange rates where two pricing kernel models are selected---perhaps of different types to reflect idiosyncrasies of the considered domestic and foreign economies. In this context, it might be of particular interest to investigate dependence structures among several pricing kernel models for all the foreign economies involved in a polyhedron of FX rates. We expect the mixing function $h(u,X)$ to play a central role in the construction of dependence models. Furthermore, a recent application by Crisafi \cite{cris} of the randomised mixtures models to the pricing of inflation-linked securities may be developed further. 

\bigskip\medskip


\noindent {\bf Acknowledgments.}\, The authors are grateful to J. Akahori, D. Brigo, D. C. Brody, C. Buescu, M. A. Crisafi, M. Grasselli, L. P. Hughston, S. Jaimungal, A. Kohatsu-Higa, O. Menoukeu Pamen, J. Sekine, W. T. Shaw and D. R. Taylor for useful comments. We would also like to thank participants at: the Actuarial Science and Mathematical Finance group meetings at the Fields Institute, Toronto, July 2011; the Fourth International MiF Conference, Kruger National Park, South Africa, August 2011; and the Mathematical Finance seminars, Department of Mathematics, Ritsumeikan University, Japan, November 2011, for helpful remarks. P. A. Parbhoo acknowledges financial support from the Commonwealth Scholarship Commission in the United Kingdom (CSC), the National Research Foundation of South Africa (NRF) and the Programme in Advanced Mathematics of Finance at the University of the Witwatersrand.


\end{document}